\documentclass[
aps,%
10pt,%
final,%
notitlepage,%
oneside,%
onecolumn,%
nobibnotes,%
nofootinbib,
superscriptaddress,%
noshowpacs,%
centertags]%
{revtex4}

\usepackage{color}

\makeatletter
\renewcommand\p@subsection{}
\makeatother

\newtheorem{theorem}{Theorem}[section]        
\newtheorem{lemma}[theorem]{Lemma}             
\newtheorem{definition}[theorem]{Definition}   
\newtheorem{proposition}[theorem]{Proposition} 

\begin{document}

\selectlanguage{english} 

\title{Secrecy results for compound wiretap channels}

\author{\firstname{Igor} \surname{ Bjelakovi\'c}}
\email{igor.bjelakovic@tum.de}
\affiliation{Lehrstuhl f\"ur theoretische Informationstechnik, Technische Universit\"at M\"unchen, 80290 M\"unchen, Germany}
\author{\firstname{Holger} \surname{Boche}}
\email{boche@tum.de}
\affiliation{Lehrstuhl f\"ur theoretische Informationstechnik, Technische Universit\"at M\"unchen, 80290 M\"unchen, Germany}
\author{\firstname{Jochen} \surname{Sommerfeld}}
\email{jochen.sommerfeld@tum.de}
\affiliation{Lehrstuhl f\"ur theoretische Informationstechnik, Technische Universit\"at M\"unchen, 80290 M\"unchen, Germany}

\begin{abstract}
We derive a lower bound on the secrecy capacity of the compound wiretap channel with channel state
information at the transmitter which matches the general upper bound on the secrecy capacity of general
compound wiretap channels given by Liang et al. and thus establishing a full coding theorem in this
case. We achieve this with a stronger secrecy criterion and the maximum error probability criterion, and
with a decoder that is robust against the effect of randomisation in the encoding. This relieves us from
the need of decoding the randomisation parameter which is in general not possible within this
model. Moreover we prove a lower bound on the secrecy capacity of the compound wiretap channel without
channel state information and derive a multi-letter expression for the capacity in this communication scenario. 
\end{abstract}

\maketitle

\section{Introduction}

Compound wiretap channels are among the simplest non-trivial models incorporating the requirement of
security against a potential eavesdropper while at the same time the legitimate users suffer from
channel uncertainty. They may be considered therefore as a starting point for theoretical investigation
tending towards applications, for example, in wireless systems, a fact explaining an alive research
activity in this area in recent years (cf. \cite{liang}, \cite{bloch} and references therein). In this
article we give capacity results for different scenarios of channel state information under a strong
secrecy criterion and the maximum error probability criterion. In a more recent work \cite{bjela4} the authors make use
of these results to derive capacity results for arbitrarily varying wiretap channels, a more realistic
communication model, which, apart from eavesdropping, takes into account an active adversarial jamming situation.\\
In this paper we consider finite families of pairs of channels $\mathfrak{W}=\{(W_t,V_t):t=1,\ldots,
T\}$ with common input alphabet and
possibly different output alphabets. The legitimate users control $W_t$ and the eavesdropper observes the
output of $V_t$. We will be dealing with two communication scenarios. In the first one the transmitter is 
informed about the index $t$ (channel state information (CSI) at the transmitter) while in the second the
transmitter has no information about that index at all (no CSI). In both scenarios the eavesdropper knows
and the legitimate receiver does not know the channel state. This setup is a generalisation of
Wyner's \cite{wyner-wire} wiretap channel.\\
Along the way we will comment what our results look like when applied to widely used class of models of
the form $\mathfrak{W}=\{(W_t,V_s):t=1,\ldots, T, s=1,\ldots, S \}$ with $T\neq S$ which are special
cases of the model we are dealing with in this paper.\\
Our contributions are summarised as follows: In \cite{liang} a general upper bound on the capacity of
compound wiretap channel as the minimum secrecy capacity of the involved wiretap channels was given. We
prove in Section \ref{csi} that the models whose secrecy capacity matches this upper bound contain
all compound wiretap channels with CSI at the transmitter. At the same time we achieve this bound with a
substantially stronger security criterion employed already in \cite{csis96}, \cite{maurer-wolf},
\cite{cai-winter-yeung}, and \cite{devetak}. Indeed, our security proof follows closely that developed in
\cite{devetak} for single wiretap channel with classical input and quantum output. In order to achieve
secrecy we follow the common approach according to which randomised encoding is a permissible
operation. Usually , the legitimate decoder can decode the sent codeword that represents both the message
to be transmitted and the outcome of the random experiment as well. However, in the case of
compound wiretap channel with CSI at the transmitter this strategy does not work as is illustrated by an
example in Section \ref{sec:example1}. We resolve this difficulty by developing a decoding strategy which
is independent of the particular channel realisation and is insensitive to randomisation while decoding
just at the optimal secrecy rate for all channels $\{W_t: t=1,\ldots, T  \}$ simultaneously.\\
Moreover, a slight modification of our proofs allows us to determine the  capacity of the compound 
wiretap channel without CSI by a (non-computable) multi-letter expression. This is content of Section
\ref{sec:no-csi}.
We should mention, however, that the traditional proof strategy of sending the pair consisting of message
and randomisation parameter to the legitimate receiver works as well in the case where the transmitter
has no CSI. 
The lower bound on the secrecy capacity, we proofed under the strong secrecy criterion, we
  have used for parts of the secrecy results for arbitrarily varying wiretap channels in
  \cite{bjela4}. The lower bound on the secrecy capacity as well the as the multi-letter expression were
  given earlier in \cite{liang} respective in \cite{bloch} for weaker secrecy criteria but without detailed proofs.\\ 
In Section \ref{sec:example2} we give an example of compound wiretap channel such that both the set of
channels to the legitimate receiver and to the eavesdropper are convex but whose secrecy capacities with
CSI and without CSI at the transmitter are different. Indeed the former is positive while the latter is
equal to $0$.\\  
Section \ref{sec:csi-t} is devoted to the practically important model
$\mathfrak{W}=\{(W_t,V_s):t=1,\ldots, T, s=1,\ldots, S \}$ with the assumption that the transmitter has
CSI for the $T$-part but has no CSI for the $S$-part of the channel. Here again we provide a multi-letter
expression for the capacity. Additionally, we give a computable description of the secrecy capacity
 in the case where the channels to the eavesdropper are degraded versions of those to the 
legitimate receiver.\\
Our results are easily extended to arbitrary sets (even uncountable) of wiretap channels via standard
approximation techniques \cite{blackw}.

\section{Compound wiretap channels}\label{sec:b}
\subsection{Definitions}

Let $A,B,C$ be finite sets and $\theta=\{1, \ldots, T\}$ an index set. We consider two families of channels 
$W_t:A\to\mathcal{P}(B)$\footnote{$\mathcal{P}(B)$ denotes the set of probability distributions on $B$.}, 
$V_t:A\to \mathcal{P}(C)$, $t\in \theta$, which we collectively abbreviate by $\mathfrak{W}$ and call
the compound wiretap channel generated by the given families of channels. 
Here the first family represents the communication link to the legitimate receiver while the output of
the latter is under control of the eavesdropper. In the rest of the paper expressions like $W_t^{\otimes
  n}$ or $V_t^{\otimes n}$ stand for the $n$-th memoryless extension of the stochastic matrices $W_t$,
$V_t$. \\
An $(n,J_n)$ code for the compound wiretap channel $\mathfrak{W}$ consists of a stochastic encoder 
$E:\mathcal{J}_n\to \mathcal{P}(A^n)$ (a stochastic matrix) with a message set
$\mathcal{J}_n:=\{1,\ldots, J_n \}$ and a collection of mutually 
disjoint decoding sets $\{D_j\subset B^n:j\in\mathcal{J}_n  \}$. The maximum error probability of a
$(n,J_n)$ code $\mathcal{C}_n$ is given by
\begin{equation}\label{eq:1}
 e(\mathcal{C}_n):= \max_{t\in \theta}  \, \max_{j \in \mathcal{J}_n} \sum_{x^n\in A^n} E(x^n|j)
W_t^{\otimes n}(D_j^c| x^n).  
\end{equation} 
I.e. neither the sender nor the receiver have CSI.\\
If channel state information is available at the transmitter the notion of $(n,J_n)$ code is modified in
that the encoding may depend on the channel index while the decoding sets remain universal,
i.e. independent of the channel index
 $t$. 
The probability of error in (\ref{eq:1}) changes to
\begin{equation*}
e_{\textup{CSI}}(\mathcal{C}_n):= \max_{t\in \theta} \, \max_{j \in\mathcal{J}_n} \sum_{x^n\in A^n}
E_t(x^n|j)
W_t^{\otimes n}(D_j^c| x^n).
\end{equation*} 
We assume throughout the paper that the eavesdropper always knows which channel is in use.
\begin{definition}\label{code}
A non-negative number $R$ is an achievable secrecy rate for the compound wiretap channel $\mathfrak{W}$
with or without CSI respectively if there is a sequence $(\mathcal{C}_n)_{n\in\mathbb{N}}$ of $(n,J_n)$
codes such that
\[ \lim_{n\to\infty} e(\mathcal{C}_n)=0 \textrm{ resp.  } \lim_{n\to\infty} e_{\textup{CSI}}(\mathcal{C}_n)=0,\]
\[\liminf_{n\to\infty}\frac{1}{n}\log J_n\ge R , \]
and
\begin{equation}\label{eq:2} 
\lim_{n\to\infty} \max_{t\in\theta} I(J; Z_t^n)=0, 
\end{equation}
where $J$ is a uniformly distributed random variable taking values in $\mathcal{J}_n$ and $Z_t^{n}$ are
the resulting random variables at the output of eavesdropper's channel $V_t^{\otimes n}$.\\
The secrecy capacity in either scenario is given by the largest achievable secrecy rate and is denoted by
$C_S(\mathfrak{W})$ and $C_{S, CSI}(\mathfrak{W})$.
\end{definition}

\subsection{Hints on operational meaning of strong secrecy}

A weaker and widely used security criterion is obtained if we replace
(\ref{eq:2}) by $\lim_{n\to\infty}\max_{t\in\theta}\frac{1}{n}I(J; Z_t^n)=0 $. We prefer to
follow \cite{csis96}, \cite{cai-winter-yeung},  and \cite{devetak} and require the validity of
(\ref{eq:2}). A nice discussion on interrelation of several secrecy criteria is
contained in \cite{bloch}. We confine ourselves to giving some hints on the operational meaning of the
requirement (\ref{eq:2}). To this end we restrict our attention to the case where the
transmitter has no CSI in order to simplify our notation. 
The case of compound wiretap channel with CSI at the transmitter can be treated accordingly. Set
\[\varepsilon_n:=  \max_{t\in\theta} I(J; Z_t^n) \textrm{ with } \lim_{n\to\infty}\varepsilon_n=0. \]
Then Pinsker's inequality implies that
\begin{equation}\label{3}
||p_{J Z_t^n}-p_{J}\otimes p_{Z_t^n}     ||\le c \sqrt{\varepsilon_n} \quad \forall t\in\theta,  
\end{equation}
with a positive universal constant $c$, where $||\cdot||$ is the variational distance. Suppose that the
eavesdropper chooses for each $t\in\theta$ decoding sets  $\{K_{j,t}\subset C^n: j\in\mathcal{J}_n  \}$
with $C^n=\bigcup_{j\in\mathcal{J}_n}K_{j,t}$. We will lower bound the average error probability (and
consequently the maximum error probability) for every choice of the decoding rule the eavesdropper might
make. Set
\begin{equation*}
e_{\textrm{av}}(t):=\frac{1}{J_n}\sum_{j\in \mathcal{J}_n}\sum_{x^n\in A^n} E(x^n|j) 
V_t^{\otimes n}(K_{j,t}^{c}|x^n).
\end{equation*}
Then
\begin{eqnarray}\label{4}
 e_{\textrm{av}}(t)&=&   \sum_{j\in\mathcal{J}_n}p_{J Z_t^n}(\{ j \}\times K_{j,t}^{c})
                            = p_{J Z_t^n}\big( \bigcup_{j\in\mathcal{J}_n} \{ j \}\times K_{j,t}^{c}   \big)\nonumber\\
&\ge& p_{J}\otimes p_{Z_t^n}\Big(\bigcup_{j\in\mathcal{J}_n}\{ j \}\times K_{j,t}^{c}\Big)-c\sqrt{\varepsilon_n}\nonumber\\
&=& \sum_{j\in\mathcal{J}_n} p_{J}\otimes p_{Z_t^n}\left(\{ j \}\times K_{j,t}^{c}\right)-c\sqrt{\varepsilon_n}
  = \frac{1}{J_n}\sum_{j\in\mathcal{J}_n}p_{Z_t^n}(K_{j,t}^{c} )-c\sqrt{\varepsilon_n}\nonumber\\
&=& \frac{J_n-1}{J_n}-c\sqrt{\varepsilon_n}
= 1-\frac{1}{J_n}-c\sqrt{\varepsilon_n},
\end{eqnarray}
where in the first and the third line we have used the fact that the sets $\{j\}\times K_{j,t}^c$,
$j\in \mathcal{J}_n$, are mutually disjoint, the second line follows from (\ref{3}), and in the
fourth line we merely observed  that for any non-negative numbers $a_1,\ldots, a_J$ with
$\sum_{j=1}^{J}a_j=1$ we have $\sum_{j=1}^{J}(1-a_j)=J-1$. 
Consequently, the average (and hence maximum) error probability of every decoding strategy the
eavesdropper might select tends to $1$ as soon as $J_n\to \infty$. 
It should be remarked, however, that although for the vast majority of messages the eavesdropper will be
in error there is still a possibility left that she/he can decode a small fraction of them correctly.\\
As will follow from the proofs below we will have $\varepsilon_n=2^{-na}$, $a>0$, and $J_n=2^{nR}$, $R>0$, if
the secrecy capacity  is positive so that the speed of convergence in (\ref{4}) will be
exponential. 

Notice that (\ref{3}) means that the random variables $Z_t^n$ at the output of the channel to the
eavesdropper are almost independent of the random variable $J$ embodying the messages to be transmitted
to the legitimate receiver. 
Therefore it is heuristically convincing that our criterion (\ref{eq:2}) offers secrecy to some extent
for communication tasks going beyond the transmission of messages. 
To demonstrate this by an example we introduce, based on \cite{ahlsw-dueck}, the notion of
identification attack as follows. Suppose that for each fixed $t\in\theta$ and any $j\in \mathcal{J}_n$
there is a subset $K_{j,t}\subset C^n$ on the eavesdropper's output alphabet where now the sets $K_{j,t}$
need not necessarily be mutually disjoint.  
With $E:\mathcal{J}_n\to \mathcal{P}(A^n)$ being the stochastic encoder used to transmit messages to the 
legitimate receiver we can write down the identification errors of first and second kind
(cf. \cite{ahlsw-dueck} for further explanation of this code concept) for the eavesdropper's channel as
\begin{equation}\label{eq:5}  
\sum_{x^n\in A^n}E(x^n|j)V^{\otimes n}_t(K_{j,t}^{c}|x^n), 
\end{equation}
and
\begin{equation}\label{eq:6}
 \sum_{x^n\in A^n}E(x^n|i)V^{\otimes n}_t(K_{j,t}|x^n) 
\end{equation}
for $j,i\in\mathcal{J}_n$, $i\neq j$.\\
One possible interpretation of this attack, again based on \cite{ahlsw-dueck}, is that on the eavesdropper's 
side of the channel there are persons $F_1,\ldots, F_{J_n}$ observing the output of the channel. 
The sole interest of $F_j$ is whether or not the message $j$ has been sent to the legitimate
receiver. 
Thus $F_j$ performs the hypothesis test represented by  $K_{j,t}$ based on his/her knowledge of $t\in
\theta$ and (\ref{eq:5}), (\ref{eq:6}) are just the errors of the first resp. second
kind for that hypothesis test.\\  
Let us define for $j\in\mathcal{J}_n$
\begin{equation*}
 g(j,t):=\sum_{x^n\in A^n}\bigg(E(x^n|j)V^{\otimes n}_t(K_{j,t}^c|x^n)+\frac{1}{J_n
     -1}\sum_{\substack{i=1\\i\neq j}}^{J_n} E(x^n|i)V^{\otimes n}_t(K_{j,t}|x^n)  \bigg)
\end{equation*} 
which is a number in $[0,2]$.\\
Notice that if
\[ g(j,t)\ge 1-\eta \]
for some $\eta\in (0,1)$ then either
\[ \sum_{x^n\in A^n}E(x^n|j)V^{\otimes n}_t (K_{j,t}^{c}|x^n)\ge \frac{1-\eta}{2}, \]
or there is at least one $i\neq j$ with
\[ \sum_{x^n\in A^n}E(x^n|i)V^{\otimes n}_t (K_{j,t}|x^n)\ge \frac{1-\eta}{2}, \]
or both, so that no reliable identification of message $j$ can be guaranteed. We show now that under
assumption of (\ref{eq:2}) we have 
\begin{equation}\label{eq:7}
 \frac{1}{J_n}\sum_{j=1}^{J_n}g(j,t)\ge 1- \eta_n, \quad \eta_n=o(n^0) 
\end{equation}
so that at most a fraction $\frac{2}{3}(1+\eta_n)$ of $j\in\mathcal{J}_n$ can satisfy the inequality
\[ g(j,t)<\frac{1}{2}.  \]
This last assertion is readily seen from (\ref{eq:7}) by applying Markov's inequality to the set
\[ F:=\{j\in \mathcal{J}_n: 2-g(j,t)>\frac{3}{2}  \}. \]
In order to prove (\ref{eq:7}), note that for any $t\in\theta$
\begin{eqnarray*}
 \frac{1}{J_n}\sum_{j=1}^{J_n}g(j,t)&=& \sum_{j=1}^{J_n}\left(p_{JZ_t^n}( \{j\}\times K_{j,t}^c)+
                                         \frac{1}{J_n -1} p_{JZ_t^n}(\{ j\}^c\times K_{j,t})
                                          \right)\\
&=& p_{JZ_t^n}(\bigcup_{j\in\mathcal{J}_n} \{j\}\times K_{j,t}^c  )+
          \frac{1}{J_n -1}\sum_{j=1}^{J_n}p_{JZ_t^n}(\{ j\}^c\times K_{j,t})\\
&\ge&  p_{J}\otimes p_{Z_t^n}(\bigcup_{j\in\mathcal{J}_n} \{j\}\times K_{j,t}^c  )+
          \frac{1}{J_n -1}\sum_{j=1}^{J_n}p_{J}\otimes p_{Z_t^n}(\{ j\}^c\times K_{j,t})-c\sqrt{\varepsilon_n}-
              c\frac{J_n}{J_n -1}\sqrt{\varepsilon_n}\\
&=& \frac{1}{J_n}\sum_{j=1}^{J_n}\left( p_{Z_t^n}(K_{j,t}^c )+ p_{Z_t^n}(K_{j,t})   \right)-c\sqrt{\varepsilon_n}\frac{2J_n-1}{J_n-1}\\
&=& 1-c\sqrt{\varepsilon_n}\frac{2J_n-1}{J_n-1},
\end{eqnarray*}
where in the third line we have used (\ref{3}) and in the fourth we inserted $p_{J}(\{ j
\}^c)=\frac{J_n-1}{J_n}$.\\ 
Besides the attempts of the eavesdropper to decode or identify messages we can introduce attacks
corresponding to each communication task introduced in \cite{ahlsw1}. 
It would be interesting, not only from the mathematical point of view, to see against which of them and to what extent secrecy can be guaranteed by the condition (\ref{eq:2}).

\section{Capacity results}\label{c}

\subsection{Preliminaries}

In what follows we use the notation as well as some properties of \emph{typical} and \emph{conditionally
  typical} sequences from \cite{csiskorn1}. For $p\in\mathcal{P}(A)$, $W:A\to\mathcal{P}(B)$, $x^n\in A^n$,
and $\delta>0$  we denote by $\mathcal{T}_{p,\delta}^n$ the set of typical sequences and by
$\mathcal{T}_{W,\delta}^n(x^n) $ the set of conditionally typical sequences given $x^n$ in the sense of \cite{csiskorn1}.\\
The basic properties of these sets that are needed in the sequel are summarised in the following three lemmata.
\begin{lemma}\label{typical}
Fixing $\delta > 0$, for every $p \in \mathcal{P}(A)$ and  $W:A \to \mathcal{P}(B)$ we have
\begin{eqnarray*}
p^{\otimes n}(\mathcal{T}_{p,\delta}^n) & \geq &1- (n+1)^{|A|} 2^{-nc\delta^2} \\
W^{\otimes n}(\mathcal{T}_{W,\delta}^n(x^n)|x^n) & \geq &1-(n+1)^{|A||B|} 2^{-nc \delta^2}
\end{eqnarray*}
for all $x^n\in A^n$ with $c=1/(2\ln 2)$. In particular, there is $n_0\in\mathbb{N}$ such that for each
$\delta>0$ and $p\in \mathcal{P}(A)$, $W:A\to\mathcal{P}(B)$ and $n>n_0$
\begin{eqnarray*}
 p^{\otimes n}(\mathcal{T}_{p,\delta}^n) & \geq &1- 2^{-nc'\delta^2} \\
W^{\otimes n}(\mathcal{T}_{W,\delta}^n(x^n)|x^n) & \geq &1- 2^{-nc' \delta^2}
\end{eqnarray*}
holds with $c'=\frac{c}{2}$.
\end{lemma}
\begin{proof}
 Standard Bernstein-Sanov trick using the properties of types from  \cite{csiskorn1} and Pinsker's inequality. 
The details can be found in \cite{wyrem} and references therein for example.
\end{proof}
Recall that for $p\in\mathcal{P}(A)$ and $W:A\to\mathcal{P}(B)$, $pW\in\mathcal{P}(B)$ denotes the output
distribution generated by $p$ and $W$ and that  $x^n \in \mathcal{T}^n_{p,\delta}$ and $y^n \in
\mathcal{T}^n_{W,\delta}(x^n)$ imply that $y^n \in \mathcal{T}^n_{pW,2|A|\delta}$. 
\begin{lemma}\label{alpha-beta} 
Let $x^n\in \mathcal{T}^n_{p,\delta}$, then for $V:A\to\mathcal{P}(C)$
\begin{eqnarray*}
|\mathcal{T}_{pV,2|A|\delta}^n| &\leq& \alpha^{-1}\\
V^n(z^n|x^n) &\leq& \beta \quad \textrm{for all} \quad z^n \in \mathcal{T}^n_{V,\delta}(x^n)
\end{eqnarray*} 
hold where
\begin{eqnarray}
\alpha &=&2^{-n(H(pV)+f_1(\delta))}\label{eq:8}\\
\beta &=&2^{-n(H(V|p)-f_2(\delta))}
\end{eqnarray}
with universal $f_1(\delta),f_2(\delta)>0$ satisfying
$\lim_{\delta\to 0}f_1(\delta)=0=\lim_{\delta\to 0}f_2(\delta)$. 
\end{lemma}
\begin{proof}
Cf. \cite{csiskorn1}.
\end{proof}
In addition we need a further lemma which will be used to determine the rates at which reliable
transmission to the legitimate receiver is possible.
\begin{lemma}\label{output-bound}
Let $p, \tilde{p} \in \mathcal{P}(A)$ and two stochastic matrices $W, \widetilde{W}:A \to \mathcal{P}(B)$
be given. Further let $q \in \mathcal{P}(B)$ be the output distribution generated by $p$ and
$W$. Fix $\delta \in (0,\frac{1}{4|A||B|})$. Then for every $n \in \mathbb{N}$
\begin{equation*}
q^{\otimes n}(\mathcal{T}^n_{\widetilde{W}, \delta}(\tilde{x}^n)) \leq (n+1)^{|A||B|}
  2^{-n(I(\tilde{p},\widetilde{W})-f(\delta))}
\end{equation*}
for all $\tilde{x}^n \in \mathcal{T}^n_{\tilde{p},\delta}$
holds for a universal $f(\delta) >0$ and $\lim_{\delta\to 0} f(\delta)=0$.
\end{lemma}
\begin{proof}
The proof can be found in \cite{wyrem} but is given here for the sake of completeness. 
Let $\tilde{x}^n \in \mathcal{T}^n_{\tilde{p},\delta}$ and $y^n \in \mathcal{T}^n_{\widetilde{W},
  \delta}(\tilde{x}^n)$. Then with the empirical distribution $p_{y^n}(b)=\frac{N(b|y^n)}{n}$, $b \in B$
it follows  by Lemma $2.6$ in \cite{csiskorn1} that
\begin{equation*}
q^n(y^n)=2^{-n(D(p_{y^n}||q)+H(p_{y^n}))} \leq 2^{-nH(p_{y^n})},
\end{equation*}
where the inequality holds, since $D(p_{y^n}||q) \geq 0$. By Lemma $2.10$ in \cite{csiskorn1}, because
$\tilde{x}^n \in \mathcal{T}^n_{\tilde{p},\delta}$ and $y^n \in \mathcal{T}^n_{\widetilde{W},
  \delta}(\tilde{x}^n)$, it follows that $y^n \in \mathcal{T}^n_{\tilde{q},2|X|\delta}$, where $\tilde{q}$
is the output distribution generated by $\tilde{p}$ and $\tilde{W}$, and thus
\begin{equation*}
\sum_{b \in B}| p_{y^n}(b)-\tilde{q}(b)| \leq 2|A||B|\delta
\end{equation*}
By the continuity of the entropy function it follows by $2.7$ in \cite{csiskorn1} that
\begin{equation*}
|H(p_{y^n})- H(\tilde{q})| \leq -2|A||B|\delta \log \frac{2|A||B|\delta}{|B|}=:\varphi(\delta)
\end{equation*}
with $\lim_{\delta\to 0} \varphi (\delta)=0$.
By the last two inequalities we obtain that 
\begin{equation}\label{eq:result}
q^n( \mathcal{T}^n_{\widetilde{W}, \delta}(\tilde{x}^n))
\leq | \mathcal{T}^n_{\widetilde{W}, \delta}(\tilde{x}^n) | 2^{-n(H(\tilde{q}) -\varphi(\delta))}.
\end{equation}
By the proof of Lemma $2.13$ it follows that
\begin{equation*}
| \mathcal{T}^n_{\widetilde{W}, \delta}(\tilde{x}^n) | \leq (n+1)^{|A||B|} 2^{n(H(\widetilde{W}|\tilde{p})+\psi(\delta))}
\end{equation*}
with $\psi(\delta)>0$ and $\lim_{\delta\to 0} \psi (\delta)=0$. Then from \eqref{eq:result} by defining
$f(\delta):=\varphi(\delta)+\psi(\delta)$ we end up with
\begin{equation*}
q^n( \mathcal{T}^n_{\widetilde{W}, \delta}(\tilde{x}^n))
\leq (n+1)^{|A||B|} 2^{-n(I(\tilde{p},\widetilde{W})-f(\delta))}
\end{equation*}
The assertion still holds if we replace $\widetilde{W}$ by $W$ and $\tilde{p}$ by $p$ throughout the
proof.
\end{proof}
The last lemma is a standard result from large deviation theory.
\begin{lemma}\label{chernoff}(Chernoff-Hoeffding bounds)
Let $Z_1,\ldots,Z_L$ be i.i.d. random variables with values in $[0,1]$ and expectation
$\mathbb{E}Z_i=\mu$, and $0<\epsilon<\frac{1}{2}$. Then it follows that
\begin{equation*}
Pr \left\{ \frac{1}{L} \sum^L_{i=1} Z_i \notin [(1\pm\epsilon)\mu]  \right\} \leq 2\exp \left( -L\cdot
\frac{\epsilon^2\mu}{3} \right), \nonumber
\end{equation*}
where $[(1\pm\epsilon)\mu]$ denotes the interval $[(1 - \epsilon)\mu, (1+ \epsilon)\mu]$.
\end{lemma}
\begin{proof}
The proof is given in \cite{dubhapanc} (cf. Theorem $1.1$) and in \cite{ahlsw-winter}.
\end{proof}

\subsection{CSI at the transmitter}\label{csi}

First we consider the case in which the transmitter has full knowledge of the channel state (CSI) while the
legitimate receiver has no information about the channel state. The main result in this section is the following theorem.
\begin{theorem}\label{CSI-code}
The secrecy capacity of the compound wiretap channel $\mathfrak{W}$ with CSI at the transmitter is given by
\begin{equation*}
C_{S,CSI}(\mathfrak{W})= \min_{t \in \theta} \max_{U_t\rightarrow X_t \rightarrow (YZ)_t}(I(U_t,Y_t)-I(U_t,Z_t)).
\end{equation*}
\end{theorem}
Here $X_t$ is a random variable with probability distribution in $\mathcal{P}(A)$ and $U_t$ is an auxiliary random
variable with range equals $A$, such that $U_t, X_t, (YZ)_t$ form a Markov chain $U_t\rightarrow X_t
\rightarrow (Y Z)_t$ in this order. Then the maximum refers to all random variables satisfying the Markov
chain condition such that $X_t$ is connected with $Y_t$ respective $Z_t$ by the channels $W_t$ respective
$V_t$ for every $t \in \theta$.

Notice first that the inequality
\[ C_{S,CSI}(\mathfrak{W})\le\min_{t \in \theta} \max_{U_t\rightarrow X_t \rightarrow (YZ)_t}(I(U_t,Y_t)-I(U_t,Z_t)) \]
is trivially true since we cannot exceed the secrecy capacity of the worst wiretap channel in the family
$\mathfrak{W}$. This has been already pointed out in \cite{liang}.
The rest of this section is devoted to the proof of the achievability.\\
\begin{proof}
It suffices to prove that
$\min_{t \in \Theta}(I(X_t, Y_t) -  I(X_t,Z_t))$
for $(XYZ)_t$ as above is an achievable secrecy rate. Then we will have shown that $R=\min_{t \in
  \Theta}(I(U_t, Y_t) -  I(U_t,Z_t))$, with $U_t \to X_t \to (YZ)_t$ form a Markov chain, is an achievable
secrecy rate (cf. \cite {csiskorn1} page $411$). We choose
$p_1, \ldots, p_T \in \mathcal{P}(A)$ and define new probability distributions on $A^n$ by
\begin{equation}\label{eq:10}
p'_t(x^n):= \left \{ \begin{array}{ll}
\frac{p^{\otimes n}_t(x^n)}{p^{\otimes n}_t(\mathcal{T}^n_{p_t,\delta})} & \textrm{if $x^n \in
  \mathcal{T}^n_{p_t,\delta}$},\\ 0 &  \textrm{otherwise}, 
\end{array} \right. .
\end{equation} 
Define then for $z^n\in C^n$, $x^n \in A^n$
\begin{equation*}
\tilde{Q}_{t,x^n}(z^n)=V_t^n(z^n|x^n)\cdot\mathbf{1}_{\mathcal{T}^n_{V_t,\delta}(x^n)}(z^n)
\end{equation*}
on $C^n$. Additionally, we set for $z^n\in C^n$
\begin{equation}\label{eq:11}
\Theta'_t(z^n) =\sum_{x^n \in \mathcal{T}^n_{p_t,\delta}} p'_t(x^n)\tilde{Q}_{t,x^n}(z^n).
\end{equation}
Now let $S:=\{z^n \in C^n : \Theta'_t(z^n) \geq \epsilon\alpha_t\}$ where
$\epsilon=2^{-nc'\delta^2}$ (cf. Lemma \ref{typical}) and $\alpha_t$ is from (\ref{eq:8}) in
Lemma \ref{alpha-beta} computed with respect to $p_t$ and $V_t$. By lemma $3.2$ the support of
$\Theta'_t$ has cardinality $\leq \alpha^{-1}_t$ since for each $x^n \in \mathcal{T}^n_{p_t,\delta}$ it holds
that $\mathcal{T}^n_{V_t,\delta}(x^n) \subset \mathcal{T}^n_{p_tV_t, 2|A|\delta}$, which implies that
$\sum_{z^n \in S} \Theta_t(z^n) \geq 1-2\epsilon$, if 
\begin{eqnarray}
\Theta_t(z^n)&=&\Theta'_t(z^n)\cdot\mathbf{1}_S(z^n) \quad \textrm{and} \nonumber\\
Q_{t,x^n}(z^n)&=&\tilde{Q}_{t,x^n}(z^n) \cdot \mathbf{1}_S(z^n).\label{eq:12}
\end{eqnarray}
Now for each $t \in \theta$ define $ J_n  \cdot L_{n,t} $ i.i.d. random
variables $X^{(t)}_{jl}$ with $j \in [J_n]:= \{1,\ldots, J_n   \}$ and $l \in [L_{n,t}]:=\{1,\ldots,
L_{n,t}  \}$ each of them distributed according to $p'_t $ with 
\begin{eqnarray}
J_n&=& \left\lfloor2^{n[\min_{t \in \theta}(I(p_t,W_t)-I(p_t,V_t))-\tau]}\right\rfloor\label{eq:13} \\
L_{n,t}&=&\left\lfloor 2^{n[I(p_t,V_t)+\frac{\tau}{4}]}\right\rfloor\label{eq:14}
\end{eqnarray} 
for $\tau>0$. Moreover we suppose that the random matrices $\{X^{(t)}_{j,l}\}_{j\in[J_n],l\in[L_{n,l}]}  $ and
$ \{X^{(t')}_{j,l}\}_{j\in[J_n],l\in[L_{n,l}]}   $ are independent for $t\neq t'$.
Now it is obvious from (\ref{eq:11}) and the definition of the set $S$ that for any $z^n\in S$
$\Theta_t(z^n)=\mathbb{E}Q_{t,X^{(t)}_{jl}}(z^n)\ge \epsilon \alpha_t$ if $\mathbb{E}$ is the expectation
value with respect to the distribution $p'_t$. For the random variables $\beta^{-1}_t
Q_{t,X^{(t)}_{jl}}(z^n)$ define the event
\begin{equation}\label{eq:15}
\iota_j(t)=\bigcap_{z^n \in C^n} \left\{\frac{1}{L_{n,t}}\sum_{l=1}^{L_{n,t}} Q_{t,X^{(t)}_{jl}}(z^n) \in [(1 \pm
  \epsilon) \Theta_t(z^n)]\right \},
\end{equation}
and keeping in mind that  $\Theta_t(z^n) \geq \epsilon \alpha_t$ for all $z^n \in S$ it follows that
for all $j \in [J_n]$ and for all $t \in \theta$
\begin{equation}\label{eq:16}
\textrm{Pr}\{ (\iota_j(t))^c\} \leq 2 |C|^n
\exp  \Big(- L_{n,t} \frac{ 2^{-n[I(p_t,V_t)+g(\delta)]}}{3}   \Big)
\end{equation}
by Lemma \ref{chernoff}, Lemma \ref{alpha-beta}, and our choice $\epsilon=2^{-nc'\delta^2}$ with 
$g(\delta):=f_1(\delta)+f_2(\delta)+3c'\delta^2$. Making $\delta>0$ sufficiently small we have for all 
sufficiently large $n\in \mathbb{N}$
\[ L_{n,t} 2^{-n[I(p_t,V_t)+g(\delta)]}\ge 2^{n\frac{\tau}{8}}. \]
Thus, for this choice of $\delta$ the RHS of (\ref{eq:16}) is double exponential in $n$ uniformly in
$t\in\theta$   and can be made smaller than $\epsilon J_n^{-1}$ for all $j \in [J_n]$ and all
sufficiently large $n \in \mathbb{N}$. I.e. 
\begin{equation}\label{eq:17}
 \textrm{Pr}\{ (\iota_j(t))^c\} \leq \epsilon J_n^{-1} \quad \forall t\in\theta.
\end{equation} 
Let us turn now to the coding part of the problem.
Let $p'_t\in\mathcal{P}(A^n)$ be given as in \eqref{eq:10}. We abbreviate
 $\mathcal{X}:=\{X^{(t)}\}_{t\in \theta}$ for the family of random
matrices $X^{(t)}=\{ X_{jl}^{(t)} \}_{j\in [J_n], l\in [L_{n,t}]}$ whose components are i.i.d. according to
$p'_t$. We will show now how the reliable transmission of the message $j\in [J_n]$ can be achieved when
randomising over the index $l\in L_{n,t}$ without any attempt to decode the randomisation parameter at
the legitimate receiver (see section \ref{sec:example1}). To this end let us define for each $j\in[J_n]$ a random set
\begin{equation*}
 D'_j(\mathcal{X}):=\bigcup_{s\in\theta}\bigcup_{k\in[L_{n,s}]} \mathcal{T}_{W_s,\delta}^{n}(X_{jk}^{(s)}),
\end{equation*} 
and the subordinate random decoder $\{ D_j(\mathcal{X}) \}_{j\in[J_n]} \subseteq B^n$ is given by
\begin{equation}\label{eq:18}
 D_j(\mathcal{X}):=D'_j(\mathcal{X})\cap \bigg( \bigcup_{\substack{j'\in [J_n] \\ j'\neq j}}
   D'_{j'}(\mathcal{X}) \bigg)^c. 
\end{equation} 
Consequently we can define the random average probabilities of error for a specific channel $t\in \theta$ by
\begin{equation}\label{eq:19}
\lambda_n^{(t)}(\mathcal{X}):=\frac{1}{J_n}\sum_{j\in [J_n]}\frac{1}{L_{n,t}}\sum_{l\in[L_{n,t}]} 
  W_t^{\otimes n}((D_j(\mathcal{X}) )^{c}| X_{jl}^{(t)} ) .
\end{equation}
Now (\ref{eq:18}) implies for each $t\in \theta$ and $l\in [L_{n,t}]$
\begin{equation}\label{eq:20}
\begin{split}
 W_t^{\otimes n}& ((D_j(\mathcal{X}) )^{c}|  X_{jl}^{(t)} ) \\ 
&\le  W_t^{\otimes n}  (\bigcap_{s\in \theta } \bigcap_{k\in[L_{n,s}]}  
(\mathcal{T}_{W_s,\delta}^{n}(X_{jk}^{(s)}))^c| X_{jl}^{(t)}        ) 
+\sum_{\substack{ j'\in[J_n]\\ j'\neq j }}\sum_{s\in \theta}\sum_{k\in[L_{n,s}]}W_t^{\otimes n}(
\mathcal{T}_{W_s,\delta}^n  (X_{j'k}^{(s)})|X_{jl}^{(t)}     )  \\
&\le W_t^{\otimes n}( (\mathcal{T}_{W_t,\delta}^{\otimes n}( X_{jl}^{(t)} ))^c| X_{jl}^{(t)}  ) 
 +\sum_{\substack{ j'\in[J_n]\\ j'\neq j }}\sum_{s\in \theta}\sum_{k\in[L_{n,s}]}W_t^{\otimes n}(
\mathcal{T}_{W_s,\delta}^n  (X_{j'k}^{(s)})|X_{jl}^{(t)}     ),
\end{split}
\end{equation}
where the second inequality follows by the monotonicity of the probability. By Lemma \ref{typical} and
the independence of all involved random variables we obtain
\begin{equation}\label{eq:21}
\begin{split}
 &\mathbb{E}_{\mathcal{X}}  (  W_t^{\otimes n}((D_j(\mathcal{X}) )^{c}| X_{jl}^{(t)} )) \\
 & \leq (n+1)^{|A||B|}\cdot 2^{-nc\delta^2} \\
 & + \sum_{\substack{ j'\in[J_n]\\ j'\neq j }}\sum_{s\in \theta}\sum_{k\in[L_{n,s}]} \mathbb{E}_{X_{j'k}^{(s)}}
      \mathbb{E}_{X_{jl}^{(t)}}  W_t^{\otimes n}( \mathcal{T}_{W_s,\delta}^n  (X_{j'k}^{(s)})|X_{jl}^{(t)}     ).
\end{split}
\end{equation} 
We shall find now for $j'\neq j$ an upper bound on
\begin{equation}\label{eq:22}
\begin{split}
\mathbb{E}_{X_{jl}^{(t)}} &  W_t^{\otimes n}( \mathcal{T}^n_{W_s,\delta}
(X_{j'k}^{(s)})|X_{jl}^{(t)}     ) \\
 &= \sum_{x^n\in A^n} p'_t(x^n)  W_t^{\otimes n}( \mathcal{T}_{W_s,\delta}^n (X_{j'k}^{(s)})|x^n  ) \\
&\le \sum_{x^n\in A^n}\frac{p_t^{\otimes n} (x^n)}{p_t^{\otimes n}(\mathcal{T}_{p_t,\delta}^n)}
W_t^{\otimes n}( \mathcal{T}_{W_s,\delta}^n (X_{j'k}^{(s)})|x^n  )\\
&=  \frac{q_t^{\otimes n}(\mathcal{T}_{W_s,\delta}^n 
(X_{j'k}^{(s)})  )  }{ p_t^{\otimes n}(\mathcal{T}_{p_t,\delta}^n)}.
\end{split}
\end{equation}
By Lemma \ref{typical} and by Lemma \ref{output-bound} for any $t,s \in \theta$ we have
\begin{equation}\label{eq:23}
\begin{split}
p_t^{\otimes n}(\mathcal{T}_{p_t,\delta}^n) & \ge 1-(n+1)^{|A|}\cdot 2^{-nc\delta^2} \\
q_t^{\otimes n}(\mathcal{T}_{W_s,\delta}^n 
(X_{j'k}^{(s)})  ) & \le (n+1)^{|A||B|}\cdot 2^{-n (I(p_s,W_s) -f(\delta))} 
\end{split}
\end{equation}
with a universal $f(\delta)>0$ satisfying $\lim_{\delta\to 0}f(\delta)=0$ since $X_{j'k}^{(s)}\in 
\mathcal{T}_{p_s,\delta}^n $ with probability 1. Thus inserting this into (\ref{eq:22}) we obtain
\begin{equation*}
 \mathbb{E}_{X_{jl}^{(t)}}  W_t^{\otimes n}( \mathcal{T}_{W_s,\delta}^n (X_{j'k}^{(s)})|X_{jl}^{(t)}     )
      \le \frac{(n+1)^{|A||B|}}{1-(n+1)^{|A|}\cdot 2^{-nc\delta^2}}\cdot 2^{-n (I(p_s,W_s) -f(\delta))}
\end{equation*}
for all $s,t\in \theta$, all $j'\neq j$, and all $l\in [L_{n,t}], k\in [L_{n,s}]$. Now by defining
$\nu_n(\delta):=(n+1)^{|A||B|}\cdot 2^{-nc\delta^2}$ and $\mu_n(\delta):=1-(n+1)^{|A|}\cdot
2^{-nc\delta^2}$ thus  for each $t\in \theta$, $l\in [L_{n,t}]$, and  $j\in[J_n]$ (\ref{eq:21}) and
(\ref{eq:22}) lead to
\begin{equation}\label{eq:24}
\begin{split}
&\mathbb{E}_{\mathcal{X}}  (  W_t^{\otimes n}((D_j(\mathcal{X}) )^{c}| X_{jl}^{(t)} )) \\
& \le \nu_n(\delta)
+ \frac{(n+1)^{|A||B|}}{\mu_n(\delta)}J_n \sum_{s\in \theta} L_{n,s}  2^{-n (I(p_s,W_s) -f(\delta))} \\
&\le  \nu_n(\delta) 
+ \frac{(n+1)^{|A||B|}}{\mu_n(\delta)} J_n \sum_{s\in \theta} 2^{-n( I(p_s,W_s)-I(p_s,V_s)
  -f(\delta)-\frac{\tau}{4}  )    }\\
&\le \nu_n(\delta)  
+ \frac{(n+1)^{|A||B|}}{\mu_n(\delta)} T\cdot  J_n \cdot 
2^{-n (\min_{s\in\theta}( I(p_s,W_s)-I(p_s,V_s) ) -f(\delta)-\frac{\tau}{4}   )} \\
&\le \nu_n(\delta) 
 + \frac{(n+1)^{|A||B|}}{\mu_n(\delta)} T \cdot 2^{-n(\tau -f(\delta)-\frac{\tau}{4} )} \\
&\le \nu_n(\delta)
 + \frac{(n+1)^{|A||B|}}{\mu_n(\delta)} T \cdot 2^{-n\frac{\tau}{2}} 
\end{split}
\end{equation}
where we have used (\ref{eq:14}), (\ref{eq:13}), and we have chosen $\delta>0$ small enough to
ensure that $\tau -f(\delta)-\frac{\tau}{4}\ge \frac{\tau}{2} $. Defining $a=a(\delta,\tau):=\frac{\min\{ c\delta^2,
  \frac{\tau}{4}  \}}{2}$ we can find $n(\delta,\tau,|A|,|B|)\in\mathbb{N}$ such that  for all $n\ge n(\delta,\tau,|A|,|B|) $ 
\begin{equation*}
 \mathbb{E}_{\mathcal{X}}(  W_t^{\otimes n}((D_j(\mathcal{X}) )^{c}| X_{jl}^{(t)} ))\le T \cdot 2^{-n a}
\end{equation*} 
holds for all $t\in \theta$, $l\in [L_{n,t}]$, and $j\in [J_n]$. Consequently, for any $t\in \theta$ we obtain
\begin{equation*}
\mathbb{E}_{\mathcal{X}}  (\lambda^{(t)}_n(\mathcal{X})) \le T\cdot 2^{-na}.
\end{equation*} 
Additionally we define for any $t\in\theta$ an event
\begin{equation}\label{eq:25}
\iota_0 (t)=\{\lambda_n^{(t)}(\mathcal{X}) \leq \sqrt{T}2^{-n\frac{a}{2}}\}.
\end{equation}
Then using the Markov inequality applied to $\lambda_n^{(t)}(\mathcal{X})$ along with
(\ref{eq:25}), we obtain that
\begin{equation}\label{eq:26}
\textrm{Pr}\{ (\iota_0(t))^c\} \leq \sqrt{T}2^{-n\frac{a}{2}}.
\end{equation}
Set
\begin{equation}\label{eq:27}
\iota:=\bigcap_{t\in\theta}\bigcap_{k=0}^{J_n}\iota_k(t) 
\end{equation}
Then with (\ref{eq:17}), (\ref{eq:26}), and applying the union bound we obtain 
\begin{eqnarray*}
\textrm{Pr}\{ \iota^c\} 
&\leq &\sum_{t\in\theta}\sum_{k=0}^{J_n}\textrm{Pr}\{(\iota_k(t))^c\} \leq
 T \cdot \epsilon +T^{\frac{3}{2}}\cdot 2^{-n \frac{a}{2}} \nonumber\\
&\le &  T^2 \cdot 2^{-nc''}
\end{eqnarray*}
for a suitable positive constant $c''>0$ and all sufficiently large $n\in\mathbb{N}$.\\
Hence, we have shown that for each $t\in \theta$ there exist realisations 
$\{ (x^{(t)}_{jl})_{j\in[J_n], l\in [L_{n,t}]}: t\in\theta\}\in \iota$ of 
$\mathcal{X}$.
Now, denoting by $\| \cdot \|$ the variational distance 
\begin{equation*}
||p-q||:=\sum_{x \in A}|p(x)-q(x)|
\end{equation*} 
for $p,q \in A$, we show that the secrecy level is fulfilled uniformly in $t\in \theta$  for any particular 
$\{ (x^{(t)}_{jl})_{j\in[J_n], l\in [L_{n,t}]}: t\in\theta\}\in \iota$ .
\begin{equation}\label{eq:28}
\begin{split}
\left\| \frac{1}{L_{n,t}}\sum^{L_{n,t}}_{l=1} V^n_t(\cdot|x^{(t)}_{jl})-\Theta_t(\cdot)\right\|
 \leq &  \frac{1}{L_{n,t}}\sum^{L_{n,t}}_{l=1} \left\|
   V^n_t(\cdot|x^{(t)}_{jl})-\tilde{Q}_{t,x^{(t)}_{jl}}(\cdot)\right\| + \\
& + \left\| \frac{1}{L_{n,t}}\sum^{L_{n,t}}_{l=1} \big(\tilde{Q}_{t,x^{(t)}_{jl}}(\cdot)- Q_{t,x^{(t)}_{jl}}(\cdot)\big)
\right\| 
+ \left\| \frac{1}{L_{n,t}}\sum^{L_{n,t}}_{l=1} Q_{t,x^{(t)}_{jl}}(\cdot)-\Theta_t(\cdot)\right\| \leq 5\epsilon .
\end{split}
\end{equation}
In the first term the functions $V^n_t(\cdot |x_{jl}^{(t)})$ and $\tilde{Q}_{t,x_{jl}^{(t)} }(\cdot)$ differ if $z^n \notin
\mathcal{T}^n_{p_tV_t,2|A|\delta}$, so it makes a contribution of $\epsilon$ to the bound. In the second term
$\tilde{Q}_t$ and $Q_t$ are different for $z^n \notin S$ and because $\iota_j(t)$ and $\sum_{z^n \in S}
\Theta_t(z^n) \geq 1 - 2\epsilon$ imply that
\begin{equation*}
\frac{1}{L_{n,t}}\sum^{L_{n,t}}_{l=1} \sum_{z^n \in S} Q_{t,x^{(t)}_{jl}}(z^n) \geq 1-3\epsilon,
\end{equation*}
the second term is bounded by $3\epsilon$. The third term is bounded by $\epsilon$ which follows directly
from \eqref{eq:15}.\\
For any $\{ (x^{(t)}_{jl})_{j\in[J_n], l\in [L_{n,t}]}: t\in\theta\}\in \iota$ with the corresponding decoding sets
$\{ D_j: j\in [J_n] \}$ it follows by construction that
\begin{equation}\label{eq:29}
\frac{1}{J_n}\sum_{j\in[J_n]}\frac{1}{L_{n,t}}\sum_{l\in [L_{n,t}]}W^{\otimes n}_t(D_j^c|
 x^{(t)}_{jl})\le \sqrt{T}
 \cdot 2^{-n a'}
\end{equation} 
is fulfilled for all $t \in \theta$ with $a'>0$, which means that we have found a $(n,J_n)$ code with
average error probability tending to zero for $n \in \mathbb{N}$ sufficiently large for any channel
realisation. Now by a standard expurgation scheme we show that this still holds for the maximum error
probability. We define the set
\begin{equation}\label{eq:30}
G_t:=\{j \in J_n: \frac{1}{L_{n,t}}\sum_{l\in [L_{n,t}]}W^{\otimes n}_t(D_j^c|x^{(t)}_{jl}) \leq \sqrt{\eta}  \}
\end{equation} 
with $\eta:=\sqrt{T}\cdot 2^{-na'}$ and denote its complement as $B_t:=G_t^c$ and the union of all
complements as $B = \bigcup_{t \in \theta} B_t$. Then \eqref{eq:29} and \eqref{eq:30} imply that
\begin{equation*}
\eta \geq \frac{1}{J_n}\sum_{j\in[J_n]}  \frac{1}{L_{n,t}}\sum_{l\in [L_{n,t}]}W^{\otimes n}_t(D_j^c|
 x^{(t)}_{jl}) \geq \frac{|B_t|}{J_n} \sqrt{\eta}
\end{equation*}
for all $t \in \theta$ and by the union bound it follows that
\begin{equation*}
|B| \leq \sum_{t \in \theta} |B_t| \leq T\cdot \sqrt{\eta} \cdot J_n .
\end{equation*}
After removing all $j\in B$ (which are at most a fraction of $T^{\frac{5}{4}}
2^{-n\frac{a'}{2}}$ of $J_n$) and relabeling we obtain a new $(n,\tilde{J}_n)$ code $(E_j,D_j)_{j \in
  [\tilde{J}_n]}$ without changing the rate. The maximum error probability of the new code fulfills for
  sufficiently large $n \in \mathbb{N}$
\begin{equation*}
\max_{t \in \theta} \, \max_{j \in [\tilde{J_n}]} \frac{1}{L_{n,t}} \sum_{l \in [L_{n,t}]} W^{\otimes n}_t(D^c_j |
x^{(t)}_{jl}) \leq T^{\frac{1}{4}} \cdot 2^{-n\frac{a'}{2}}.
\end{equation*}
On the other hand, if we set
\begin{equation}\label{eq:31}
 \hat{V}^n_t(z^n|(j,l)):= V^n_t(z^n|x^{(t)}_{jl})
\end{equation}
and further define
\begin{eqnarray}
\hat{V}^n_{t,j}(z^n) & = & \frac{1}{L_{n,t}}\sum ^{L_{n,t}}_{l=1}\hat{V}^n_t(z^n|(j,l))\;  ,
\label{eq:32} \\
\bar{V}^n_t(z^n) & = & \frac{1}{\tilde{J}_n}\sum^{\tilde{J}_n}_{j=1}  \hat{V}^n_{t,j}(z^n), \label{eq:33}
\end{eqnarray}
we obtain that
\begin{eqnarray*}
\|\hat{V}^n_{t,j}-\bar{V}^n_t\| &\leq & \|  \hat{V}^n_{t,j}-\Theta_t\| +  \| \Theta_t- \bar{V}^n_t\|\nonumber\\
&\le & 10 \epsilon,
\end{eqnarray*}
for all  $j\in [\tilde{J}_n], t\in \theta$ with $\epsilon=2^{-nc'\delta^2}$ where we have used the convexity
of the variational distance and \eqref{eq:28} which still applies by our expurgation procedure. For a uniformly distributed
random variable $J$ taking values in the set $\{1, \ldots , \tilde{J}_n\}$ we
obtain with Lemma $2.7$ of \cite{csiskorn1} (uniform continuity of the entropy function)
\begin{eqnarray*}
I(J;Z^n_t) &=& \sum^{J_n}_{j=1} \frac{1}{\tilde{J}_n} (H(\bar{V}_t^n)-H(\hat{V}_{t,j}^n))\\
               &=& H(Z^n_t)-H(Z^n_t|J)\\
                &\leq&-10\epsilon\log(10\epsilon)+10n\epsilon\log|C|
\end{eqnarray*}
uniformly in $t\in \theta$ (for $10\epsilon\leq e^{-1}$). Hence
the strong secrecy level of the definition \ref{code} holds uniformly in $t\in\theta$.
Using standard arguments (cf. \cite{csiskorn1} page $411$) we then have shown the achievability of
the secrecy rate
\begin{equation}\label{eq:34}
R_S=\min_{t \in \theta} \max_{U_t\rightarrow X_t \rightarrow (YZ)_t}(I(U_t,Y_t)-I(U_t,Z_t)).
\end{equation}
\end{proof}
\emph{Remark.} Note that in the case that $\mathfrak{W}:=\{ W_t,V_s: t=1,\ldots T, s=1,\ldots S\}$ with
$S\neq T$ and the pair $(s,t)$ known to the transmitter prior to transmission nothing new happens. A
slight modification of the arguments presented above shows that
\[C_{S,CSI}(\mathfrak{W})= \min_{(t,s)} \max_{U\rightarrow X \rightarrow (Y_t Z_s)}(I(U,Y_t)-I(U,Z_s)). \]

\subsection{No CSI}\label{sec:no-csi}

In the previous section we have assumed that the channel state is known to the transmitter. We now
consider the case where neither the transmitter nor the receiver has knowledge of the channel state. We
will prove that 
\begin{theorem}\label{low-bound}
For the secrecy capacity $C_S(\mathfrak{W})$ of the compound wiretap channel $\mathfrak{W}$ without CSI
it holds that
\begin{equation*}\label{eq:ratebound}
 C_S(\mathfrak{W})\ge \max_{p\in\mathcal{P}(A)}( \min_{t \in \theta}I(p,W_t)-\max_{t \in \theta}I(p,V_t)). 
\end{equation*}
\end{theorem}
\smallskip
\begin{proof}
Caused by the lack of channel knowledge we use a stochastic encoder independent of the channel
realisation. For any $p\in\mathcal{P}(A)$ let $p' \in \mathcal{P}(A^n)$ be the distribution given by
\begin{equation*}
p'(x^n):= \left \{ \begin{array}{ll}
\frac{p^{\otimes n}(x^n)}{p^{\otimes n}(\mathcal{T}^n_{p,\delta})} & \textrm{if $x^n \in \mathcal{T}^n_{p,\delta}$},\\
0 &  \textrm{otherwise}.
\end{array} \right.
\end{equation*}
Then  analogously to the case with CSI we define $\tilde{Q}_{t,x^n}(z^n),
{Q}_{t,x^n}(z^n)$, and $\Theta'_t(z^n), \Theta_t (z^n)$ for $z_n \in C^n$ but now
with respect to the distribution $p'$. Consequently, $\Theta'(\cdot)$ has support only on
$\mathcal{T}^n_{pV_t,2|A|\delta}$, and ${Q}_{t,x^n}(\cdot)$ and $\Theta(\cdot)$ only on the
set $S$. Furthermore $\Theta(z^n) \geq \epsilon \alpha_t$ for all $z^n \in S$. Now define $J_n\cdot L_n$
i.i.d random variables $X_{jl}$ according to the distribution $p'$ independent of $t \in \theta$ with
$j\in [J_n]$ and $l \in [L_n]$  with
\begin{eqnarray}
J_n&=& \lfloor 2^{n[\min_tI(p,W_t)-\max_tI(p,V_t)-\tau]} \rfloor \label{eq:35}\\
L_{n}&=& \lfloor 2^{n[\max_t I(p,V_t)+ \frac{\tau}{4}] } \rfloor \label{eq:36}
\end{eqnarray}
for $\tau>0$. Now because $\Theta_t(z^n)=\mathbb{E}Q_{t,X_{jl}} \geq \epsilon \alpha_t$ for all $z^n \in
S$ we define the event $\iota_j(t)$  as in \eqref{eq:15}  for the random variables $\beta_t^{-1}
Q_{t,X_{jl}}$ 
\begin{equation*}
\iota_j(t)=\bigcap_{z^n \in C^n} \left\{\frac{1}{L_{n}}\sum_{l=1}^{L_{n}} Q_{t,X_{jl}}(z^n) \in [(1 \pm
  \epsilon) \Theta_t(z^n)]\right \},
\end{equation*}
but considering the difference that the random variables $X_{jl}$ are independent of the channel
state. Then analogously to \eqref{eq:16} we obtain that
\begin{equation*}
\textrm{Pr}\{(\iota_j(t))^c\} \leq  2 |C|^n\exp 
\Big(-L_{n}\frac{2^{-n(I(p,V_t)+g(\delta))}}{3}\Big)
\end{equation*}
by Lemma \ref{chernoff} and Lemma \ref{alpha-beta}. Notice that, because the sender does not know which
channel is used, we need the maximum in the definition of $L_n$.  Thus the
right-hand side is a double exponential in $n$ and can be made smaller than $\epsilon J_n^{-1}$ for all
$j$ and for all $t\in \theta$ and sufficiently large $n$.\\
Now let $J_n$ and $L_n$ be defined as stated above, and let $X^n=\{X_{jl}\}_{j \in [J_n], l \in [L_n]}$
 be the set of i.i.d. random variables each of them distributed according to $p'$ independent of $t \in
 \theta$. As in the case of CSI we can show that reliable transmission of the message $j \in [J_n]$ can be
 achieved. To this end define now
 the random decoder $\{D_j(X^n)\}_{j \in [J_n]} \subseteq B^n$ as in \eqref{eq:18} but with
\begin{equation*}
 D'_j(X^n):=\bigcup_{s\in\theta}\bigcup_{k\in[L_{n}]} \mathcal{T}_{W_s,\delta}^{n}(X_{jk}),
\end{equation*}
and the random average probabilities of error for a specific channel $\lambda^{(t)}_n(X^n)$ as in
\eqref{eq:19}. Notice that now both $X^n$ and $L_n$ do not depend on $t \in \theta$ and this holds
throughout the entire proof. Then we can  give the bound in \eqref{eq:20} now by
\begin{equation*}
\begin{split}
W^{\otimes n}_t & ((D_j(X^n))^c | X_{jl})\\
&\leq W^{\otimes n}_t(( \mathcal{T}^{\otimes n}_{W_t, \delta}(X_{jl}))^c | X_{jl})
+\sum_{\substack{ j'\in[J_n]\\ j'\neq j }}\sum_{s\in \theta}\sum_{k\in[L_{n}]}W_t^{\otimes n}(
\mathcal{T}_{W_s,\delta}^n  (X_{j'k})|X_{jl}     )
\end{split}
\end{equation*}
We can bound the first term in the inequality by $\nu_n(\delta):= (n+1)^{|A||B|} \cdot 2^{-nc \delta^2}$
(see \eqref{eq:21}). If we average over all codebooks we get
\begin{equation*}
\begin{split}
\mathbb{E}_{X^n} &  (W^{\otimes n}_t ((D_j(X^n))^c | X_{jl}))  \\
&  \leq \nu_n(\delta) + \sum_{\substack{ j'\in[J_n]\\ j'\neq j }}\sum_{s\in
    \theta}\sum_{k\in[L_{n}]} 
\mathbb{E}_{X_{j'k}} \mathbb{E}_{X_{jl}} W_t^{\otimes n} (\theta_{W_s,\delta}^n  (X_{j'k})|X_{jl}     ).
\end{split}
\end{equation*}
By the same reasoning as in \eqref{eq:22} and \eqref{eq:23} we can give an upper bound on
\begin{equation*}
\begin{split}
 \mathbb{E}_{X_{jl}} & W_t^{\otimes n} (\mathcal{T}^n_{W_s,\delta}  (X_{j'k})|X_{jl}     )\\
& \leq \frac{q^{\otimes n}_t (\mathcal{T}_{W_s,\delta}^n  (X_{j'k}))}{p^{\otimes n} (\mathcal{T}^n_{p,\delta})}\\
&\le \frac{(n+1)^{|A||B|}}{1-(n+1)^{|A|}\cdot 2^{-nc\delta^2}}\cdot 2^{-n (I(p,W_s) -f(\delta))}
\end{split}
\end{equation*}
 for all $t \in \theta$, all $j' \neq j$ and all $k,l \in [L_n]$ with a universal $f(\delta)>0$ satisfying
 $\lim_{\delta \to 0} f(\delta)=0$. $q^{\otimes n}_t$ denotes the output distribution generated by the
 conditional distribution $W^{\otimes n}_t$ and the input distribution $p^{\otimes n}$. Additionally we
 define $\mu_n(\delta):=1-(n+1)^{|A|} \cdot 2^{-nc \delta^2}$. Then \eqref{eq:24} changes to
\begin{equation*}
\begin{split}
&\mathbb{E}_{X^n} (W^{\otimes n}_t ((D_j(X^n))^c | X_{jl}))\\
&\leq \nu_n(\delta)+\frac{(n+1)^{|A||B|}}{\mu_n(\delta)} T \cdot J_n L_n \cdot 2^{-n(\min_s
  I(p,W_s)-f(\delta))}\\
&\leq \nu_n(\delta)+ T\cdot 2^{-n\frac{\tau}{2}} 
\end{split}
\end{equation*}
by the definition of $J_n$ and $L_n$ in \eqref{eq:35}, \eqref{eq:36} and by choosing $\delta>0$ small
enough that $\tau-\frac{\tau}{4}-f(\delta) \geq \frac{\tau}{2}$. Now by defining
$a:=\frac{\min\{c\delta^2,\frac{\tau}{2}\}}{2}$ and the definition of the error probability the last
inequality results in the upper bound 
\begin{equation*}
\mathbb{E}_{X^n}(\lambda^{(t)}_n (X^n)) \leq T \cdot 2^{-na}
\end{equation*}
for any $t\in\theta$ and $n \in \mathbb{N}$ large enough. \\
Now we define the event $\iota_0(t)$ for any $t \in \theta$ and the event $\iota$ as in
\eqref{eq:25} and \eqref{eq:27} but with the difference that the input is independent of
the channel realisation. So by the same reasoning we end in
\begin{equation*}
\textrm{Pr} \{\iota^c\}\leq T^2 \cdot 2^{-nc''}
\end{equation*}
for a constant $c''>0$ and all sufficiently large $n \in \mathbb{N}$, which implies that there exist
realisations $\{x_{jl}\}$ of $\{X_{jl}\}$ such that $x_{jl} \in \iota$ for all $j \in [J_n]$ and $l \in
[L_n]$.  Then analogously to \eqref{eq:28} we get for any channel $t\in\theta$
\begin{equation*}
\left\| \frac{1}{L_n} \sum_{l=1}^{L_n} V^n_t(\cdot|x_{jl})- \Theta_t(\cdot)\right\| \leq 5 \epsilon
\end{equation*}
differs from the former only by $L_n$ in place of $L_{n,t}$.
Hence, following the same arguments subsequent to \eqref{eq:29}, we have shown that there is a
sequence of $(n,\tilde{J}_n)$ codes for which
\begin{equation*}
\max_{t \in \theta} \, \max_{j \in [\tilde{J}_n]} \frac{1}{L_n} \sum_{l \in [L_n]} W^{\otimes n}_t(D^c_j |
x_{jl}) \leq T^{\frac{1}{4}} \cdot 2^{-n\frac{a'}{2}}
\end{equation*}
holds for sufficiently large $n\in\mathbb{N}$, and the strong secrecy level is fulfilled for every channel
$t\in\theta$ by
\begin{equation*}
\|\hat{V}^n_{t,j}-\bar{V}^n_t \| \leq 6\epsilon
\end{equation*}
($\hat{V}^n_{t,j},\bar{V}^n_t$ defined as in \eqref{eq:32}, \eqref{eq:33}) and thus by
\begin{equation*}
I(J;Z^n_t) \leq -10 \epsilon\log(10\epsilon)+10n\epsilon \log|C|
\end{equation*}
which tends to zero for $n \to \infty$ uniformly in $t\in \theta$.
\end{proof} 

We turn now to the converse of Theorem \ref{low-bound}. Actually, we give only a
multiletter formula of the upper bound of the secrecy rates. First we need the following lemma.
\begin{lemma}\label{lemma-superadditivity}
 Let $\mathfrak{W}=\{ (W_t,V_t): t\in \theta   \}$ be an arbitrary compound wiretap channel without CSI. Then
\[ \lim_{n\to\infty}\frac{1}{n}\max_{U\to X^n \to Y_t^n Z_t^n}(\inf_{t\in \theta}I(U;Y_t^n)-\sup_{t \in \theta} I(U;Z_t^n)) \] 
exists and we have
\begin{eqnarray*}
 \lim_{n\to\infty}\frac{1}{n} 
 \max_{U\to X^n \to Y_t^n Z_t^n}(\inf_{t\in \theta}I(U;Y_t^n)-\sup_{t \in \theta}I(U;Z_t^n)       )\\
 = \sup_{n\in\mathbb{N}}\frac{1}{n}\max_{U\to X^n \to Y_t^n Z_t^n}(\inf_{t\in \theta} 
 I(U;Y_t^n)-\sup_{t\in \theta} I(U;Z_t^n)       ).
\end{eqnarray*}
\end{lemma}
\bigskip
\begin{proof}
The proof is based on Fekete's lemma \cite{fekete}. Consequently, if we apply the lemma to the sequence
$(a_n)_{n \in \mathbb{N}}$ defined by 
\[ a_n:= \max_{U\to X^n \to Y_t^n Z_t^n}(\inf_{t\in \theta}I(U;Y_t^n)-\sup_{t\in \theta}I(U;Z_t^n)) \]
it suffices to show that the inequality
\[ a_{n+m}\ge a_n +a_m \]
holds for all $n,m\in\mathbb{N}$. This will be done by considering two  independent Markov chains 
$U_1\to X^n\to (Y_t^n,Z_t^n)$ and $U_2\to \hat X^m\to  (\hat Y_t^m,\hat Z_t^m)$
and setting $U:=(U_1,U_2)$, $X^{n+m}:=(X^n,\hat X^m)$, and $(Y_t^{n+m},Z_t^{n+m}):= ((Y_t^n,\hat
Y_t^m),(Z_t^n,\hat Z_t^m)   )$. Then by the definition of $a_n$
\begin{equation*}
\begin{split}
a_{n+m} & \geq \inf_{t \in \theta} I(U;Y^{n+m}_t)-\sup_{t\in \theta} I(U;Z^{n+m}_t) \\
& \geq \inf_{t \in \theta} I(U_1;Y^{n}_t) +  \inf_{t \in \theta} I(U_2;\hat{Y}^{m}_t) 
 - \sup_{t\in \theta} I(U_1;Z^{n}_t) - \sup_{t\in \theta} I(U_2;\hat{Z}^{m}_t). 
\end{split}
\end{equation*}
By the independence of the two Markov chains mentioned above and because apart from that these
Markov chains were arbitrary we can conclude that
\begin{equation*}
a_{n+m} \geq a_n +a_m
\end{equation*}
holds for all $n,m \in \mathbb{N}$.
\end{proof}
\begin{proposition}\label{multiletter-converse}
The secrecy capacity of the compound wiretap channel in the case of no CSI $C_S(\mathfrak{W})$ is upper
bounded by  
\small
\begin{equation*}
C_S(\mathfrak{W}) \le \lim_{n\to\infty}\frac{1}{n}\max_{U\to X^n\to Y_t^n Z_t^n}(\inf_{t\in \theta}I(U;Y_t^n)-
\sup_{t\in \theta}I(U;Z_t^n)   ).
\end{equation*}
\end{proposition}
\begin{proof}
Let $(\mathcal{C}_n)_{n\in\mathbb{N}}$ be any sequence of $(n,J_n)$ codes such that with
\begin{equation}\label{eq:37}
 \sup_{t\in \theta}\frac{1}{J_n}\sum_{j=1}^{J_n}\sum_{x^n\in A^n}E(x^n|j)
 W_t^{\otimes n}(D_j^{c}|x^n)=:  \varepsilon_{1,n}, 
\end{equation}
and
\begin{equation*}
\sup_{t\in \theta}I(J;Z_t^n) =:  \varepsilon_{2,n}
\end{equation*} 
it holds that $ \lim_{n\to\infty}\varepsilon_{1,n}=0$ and $ \lim_{n\to\infty}\varepsilon_{2,n}=0$, where $J$ denotes
the random variable which is uniformly distributed on the message set $\{1,\ldots, J_n  \}$. Let us
denote by $\hat J$ the random variable with values in $\{1,\ldots,J_n  \}$ determined by the Markov chain
$J\to X^n\to Y_t^n\to \hat J$ where the first transition is governed by $E$, the second by $W_t^{\otimes
  n}$, and the last by the decoding rule. Then we have for any $t\in \theta$
\begin{eqnarray}\label{eq:38}
 \log J_n &=& H(J) = I(J;\hat J)+H(J|\hat J)\nonumber\\
          &\le& I(J;Y_t^n)+ H(J|\hat J), 
\end{eqnarray}
where the inequality follows from the data processing inequality. Then using Fano's inequality we find that 
\begin{equation*}
 H(J|\hat J)\le 1+ \varepsilon_{1,n}\log J_n
\end{equation*} 
with \eqref{eq:37}. Thus we can rewrite inequality \eqref{eq:38} as
\begin{equation*}
 (1-\varepsilon_{1,n})\log J_n\le I(J;Y_t^n)+1
\end{equation*}
for all $t\in \theta$.
On the other hand we have for every $t\in \theta$
\begin{equation*}
 I(J;Y_t^n) =  I(J;Y_t^n)-\sup_{t\in \theta}I(J;Z_t^n)+\varepsilon_{2,n}
\end{equation*}
where we have used the validity of the secrecy criterion stated above. Then the last two inequalities
imply that for any $t\in \theta$
\begin{equation}\label{eq:39}
 (1-\varepsilon_{1,n})\log J_n \le I(J;Y_t^n)-\sup_{t\in \theta}I(J;Z_t^n)+\varepsilon_{2,n}.
\end{equation}
Since the LHS of (\ref{eq:39}) does not depend on $t$ we arrive at
\begin{equation*}
 (1-  \varepsilon_{1,n})  \log J_n 
\le \max_{U\to X^n\to Y_t^n Z_t^n} ( \inf_{t\in \theta}I(U;Y_t^n)-\sup_{t\in \theta}I(U;Z_t^n)) +\varepsilon_{2,n},
\end{equation*} 
which concludes the proof after dividing by $n \in \mathbb{N}$, taking $\limsup$ and taking into account the assertion
of Lemma \ref{lemma-superadditivity}.
\end{proof}
\emph{Remark.} Following the same arguments subsequent to \eqref{eq:34} concerning the use of the
channels defined by $P_{Y_t|T}=W_t \cdot P_{X|T}$ and $P_{Z_t|T}=V_t \cdot P_{X|T}$ instead of $W_t$ and
$V_t$ and applying the assertion of Theorem \ref{low-bound} to the $n$-fold product of channels $W_t$ and
$V_t$, we can give the coding theorem for the multiletter case. The capacity of the compound wiretap
channel in the case of no CSI is
\begin{equation*}
C_S(\mathfrak{W}) = \lim_{n\to\infty}\frac{1}{n}\max_{U\to X^n\to Y_t^n Z_t^n}(\inf_{t\in \theta}I(U;Y_t^n)-
\sup_{t\in \theta }I(U;Z_t^n)   ).
\end{equation*}
Let us consider now the case $\mathfrak{W}:=\{ W_t,V_s: t=1,\ldots T, s=1,\ldots S\}$ with $S\neq T$ and
the pair $(s,t)$ unknown to both the transmitter and the legitimate receiver. Additionally we assume that
each $V_s$ is a degraded version of every $W_t$, which is characterised by 
\begin{equation}\label{eq:40}
 V_s(z|x)=\sum_{y\in B}W_t(y|x)D_{(t,s)}(z|y),
\end{equation}
for all $x\in A$, $z\in C$, if $D$ is defined as the stochastic matrix $D:B\to \mathcal{P}(C)$. Then we
have the following 
\begin{lemma}\label{concavity}
Let $p\in\mathcal{P}(A)$, $W:A\to\mathcal{P}(B)$, $V:A\to\mathcal{P}(C)$, and assume that $V$ is a
degraded version of $W$.
Then $I(X;Y|Z)$ is a concave with respect to the input distribution $p_X=p$.
\end{lemma}
\begin{proof}
Let $X,Y,Z$ be random variables with values in $A,B,C$ respectively distributed according to
\begin{equation}\label{eq:41}
\textrm{Pr}(X=x, Y=y,Z=z)  := p_{XYZ}(x,y,z)=p(x)W(y|x)D(z|y)
\end{equation}
for all $x \in A,y \in B,z \in Z$. Because
\begin{equation*}
 I(X;Y|Z)=H(Y|Z)-H(Y|X,Z)
\end{equation*}
the proof is based on the two assertions
\begin{enumerate}
 \item $H(Y|Z)$ depends concavely on $p_X$, and
 \item $H(Y|X,Z)$ is an affine function of $p_X$.
\end{enumerate}
\smallskip
First, $H(Y|Z)$ is a concave function with respect to $p_{YZ}$ by the log-sum inequality
(cf. \cite{csiskorn1} Lemma 3.1). Then because $p_{XYZ}$ depends affinely on $p_X$ by \eqref{eq:41},
so does $p_{YZ}$, and the first assertion follows. For the second consider that \eqref{eq:40} and
\eqref{eq:41} imply that
\begin{equation*}
p_{Y|X,Z}(y|x,z)=\frac{W(y|x)D(z|y)}{V(z|x)}
\end{equation*}
for every input distribution $p_X$, any $y \in B$ and all $x \in A, z \in C$ with $p_{XZ}(x,z) >0$. Then
we have
\begin{equation*}
H(Y|X,Z)= \sum_{x \in A,z \in C} p_{XZ}(x,z)H\left( \frac{W(\cdot|x)D(z|\cdot)}{V(z|x)} \right)
\end{equation*}
showing that $H(Y|X,Z)$ is an affine function of $ p_{XZ}$ which in turn depends affinely on $p_X$. 
\end{proof}
Now because the random variables $X,Y_t,Z_s$ ($Y_t,Z_s$ the channel outputs of $W_t$ and $V_s$ resp.)
form a Markov chain for all $t \in \theta$ and $s \in \mathcal{S}$, we obtain that
\begin{equation}\label{eq:42}
I(X;Y_t|Z_s)=I(X,Y_t)-I(X,Z_s).
\end{equation}
By virtue of Theorem 2 of \cite{ahlswcsis1} we can show that for the secrecy rate it holds that
\begin{equation*}
R_S \leq \frac{1}{n} \sum^n_{i=1}I(X_i;Y_{i,t}|Z_{i,s})+\epsilon'
\end{equation*}
for any channel $(t,s)\in \theta \times \mathcal{S}$ and $\epsilon'>0$. The concavity of $I(X;Y_{t}|Z_{s})$
with respect to the input distributions $p \in \mathcal{P}(A)$ together with \eqref{eq:42} then imply
the converse part of Theorem \ref{low-bound}, that
\begin{equation*}
R_S \leq \max_{p\in \mathcal{P}(A)}\min_{(t,s)}(I(p,W_t)-I(p,V_s)).
\end{equation*}
Now we can state the following
\begin{proposition}
If $V_s$ is a degraded version of $W_t$ for all $s \in \mathcal{S}$ and $t \in \theta$ the capacity of the
compound wiretap channel is given by
\begin{eqnarray*}
C_S(\mathfrak{W})&=& \max_{p\in\mathcal{P}(A)}\min_{(t,s)}(I(p,W_t)-I(p,V_s))\\
                  &=& \max_{p\in\mathcal{P}(A)}(\min_{t}I(p,W_t)-\max_{s}I(p,V_s)).
\end{eqnarray*}
\end{proposition}
\smallskip
\emph{Remark.} This result was obtained in \cite{liang} with a weaker notion of secrecy.

\subsection{Channel state to the legitimate receiver is known at the transmitter ($CSI_t$)}\label{sec:csi-t}

We now consider the case, in which the transmitter has knowledge of the channel state to the
legitimate receiver $t \in \theta$ but the channel state to the eavesdropper $s \in \mathcal{S}$ is
unknown. We will denote this  kind of channel state information by $\text{CSI}_t$. Consequently we get
for each $t \in \theta$ possible channel realisations $\mathfrak{W}_t:=\{ (W_t,V_s): s=1,\ldots
S\}$. Then we can describe the compound channel as $\mathfrak{W}=\cup_{t \in \theta} \mathfrak{W}_t$. 
\begin{theorem}\label{CSI-t}
For the secrecy capacity $C_{S,CSI_t}(\mathfrak{W})$ of the compound wiretap channel with $CSI_t$ it holds that
\begin{equation*}
C_{S,CSI_t}(\mathfrak{W}) \geq  \min_{t \in \theta} \max_{p\in\mathcal{P}(A)}(I(p,W_t)-\max_{s \in
  \mathcal{S}} I(p,V_s)).
\end{equation*}
\end{theorem}
\smallskip
\begin{proof}
Adapted to the channel realisation $W_t$ define
\begin{equation}\label{eq:43}
p'_t(x^n):= \left \{ \begin{array}{ll}
\frac{p^{\otimes n}_t(x^n)}{p^{\otimes n}_t(\mathcal{T}^n_{p_t,\delta})} & \textrm{if $x^n \in \mathcal{T}^n_{p_t,\delta}$},\\
0 &  \textrm{otherwise}.
\end{array} \right.
\end{equation}
for arbitrary input distributions $p_1, \ldots, p_T \in \mathcal{P}(A)$.
Now define for $z^n\in C^n$ and $s\ \in \mathcal{S}$
\begin{equation*}
\tilde{Q}_{s,x^n}(z^n)=V_s^n(z^n|x^n)\cdot\mathbf{1}_{\mathcal{T}^n_{V_s,\delta}(x^n)}(z^n)
\end{equation*}
on $C^n$.
Additionally, we set for $z^n\in C^n$
\begin{equation*}
\Theta'_s(z^n) =\sum_{x^n \in \mathcal{T}^n_{p_t,\delta}} p'_t(x^n)\tilde{Q}_{s,x^n}(z^n).
\end{equation*}
Now let $S:=\{z^n \in C^n : \Theta'_s(z^n) \geq \epsilon\alpha_{t,s} \}$ where
$\epsilon=2^{-nc'\delta^2}$ and $\alpha_{t,s}$ is from (\ref{eq:8}) similar to the former cases
but computed with respect to $p_t$ and $V_s$. Then the support of $\Theta'_s$ has cardinality $\leq
\alpha^{-1}_{t,s}$, which implies that $\sum_{z^n \in S} \Theta_s(z^n) \geq 1-2\epsilon$. Analogously to
\eqref{eq:12} define $\Theta_s(z^n)$ and $Q_{s,x^n}(z^n)$ with
support on $S$ and further
\begin{eqnarray}
J_n&=& \lfloor 2^{n[\min_t (I(p_t,W_t)-\max_sI(p_t,V_s))-\tau]} \rfloor \label{eq:44}\\
L_{n,t}&=& \lfloor 2^{n[\max_s I(p_t,V_s)+ \frac{\tau}{4}] } \rfloor .\label{eq:45}
\end{eqnarray}
As in the case of CSI define random matrices $\{X^{(t)}_{jl}\}_{j\in[J_n],l\in[L_{n,t}]}$ such that the
random variables $X^{(t)}_{jl}$ where i.i.d. according to $p'_t$. We suppose additionally that
$\{X^{(t)}_{jl}\}_{j,l}$ and $\{X^{(t')}_{jl}\}_{j,l}$ are independent for $t\neq t'$. For any $z^n\in S$
it follows that $\Theta_s(z^n)=\mathbb{E}Q_{s,X^{(t)}_{jl}}(z^n)\ge \epsilon \alpha_{t,s}$, if $\mathbb{E}$ is the
expectation value with respect to the distribution $p'_t$. For the random variables $\beta^{-1}_{t,s}
Q_{s,X^{(t)}_{jl}}(z^n)$ define the event
\begin{equation*}
\iota_j(s,t)=\bigcap_{z^n \in C^n} \left\{\frac{1}{L_{n,t}}\sum_{l=1}^{L_{n,t}} Q_{s,X^{(t)}_{jl}}(z^n) \in [(1 \pm
  \epsilon) \Theta_s(z^n)]\right \}.
\end{equation*}
Then it follows that for all $j \in [J_n]$ and for all $s \in \mathcal{S}$ it holds for each $t \in \theta$
\begin{equation*}
\textrm{Pr}\{ (\iota_j(s,t))^c\} \leq 2 |C|^n
\exp  \Big(- L_{n,t} \frac{ 2^{-n[I(p_t,V_s)+g(\delta)]}}{3}   \Big)
\end{equation*}
by Lemma \ref{chernoff}, Lemma \ref{alpha-beta}, Thus the RHS is double exponential in $n$ uniformly in
$s\in\mathcal{S}, t \in \theta$ (guaranteed by the maximum in $s$ in the definition of $L_{n,t}$) and can
be made smaller than $\epsilon J_n^{-1}$ for all $j \in [J_n]$ and all sufficiently large $n$. Now the
coding part of the problem is similar to the case with CSI. Let $p'_t\in\mathcal{P}(A^n)$ be given as in
\eqref{eq:43}. We abbreviate  $\mathcal{X}:=\{X^{(t)}\}_{t\in \theta}$ for the family of random
matrices $X^{(t)}=\{ X_{jl}^{(t)} \}_{j\in [J_n], l\in [L_{n,t}]}$ whose components are i.i.d. according to
$p'_t$. We will show how reliable transmission of the message $j\in [J_n]$ can be achieved. To
this end define now the random decoder $\{D_j(\mathcal{X})\}_{j \in [J_n]} \subseteq B^n$ as in
\eqref{eq:18} and with
\begin{equation*}
 D'_j(\mathcal{X}):=\bigcup_{r\in\theta}\bigcup_{k\in[L_{n,r}]} \mathcal{T}_{W_r,\delta}^{n}(X_{jk}^{(r)}),
\end{equation*}
and the random average probabilities of error for a specific channel $\lambda^{(t)}_n(\mathcal{X})$ as in
\eqref{eq:19} by
\begin{equation*}
\lambda_n^{(t)}(\mathcal{X}):=\frac{1}{J_n}\sum_{j\in [J_n]}\frac{1}{L_{n,t}}\sum_{l\in[L_{n,t}]} 
  W_t^{\otimes n}((D_j(\mathcal{X}) )^{c}| X_{jl}^{(t)} ) .
\end{equation*}
As in \eqref{eq:20} we get for each $t\in \theta$ and $l\in [L_{n,t}]$
\begin{equation*}
\begin{split}
 W_t^{\otimes n}&((D_j(\mathcal{X}) )^{c}| X_{jl}^{(t)} )  \\
&\le W_t^{\otimes n}( (\mathcal{T}_{W_t,\delta}^{\otimes n}( X_{jl}^{(t)} ))^c| X_{jl}^{(t)}  ) 
 +\sum_{\substack{ j'\in[J_n]\\ j'\neq j }}\sum_{r\in \theta}\sum_{k\in[L_{n,r}]}W_t^{\otimes n}(
\mathcal{T}_{W_r,\delta}^n  (X_{j'k}^{(r)})|X_{jl}^{(t)}     ),
\end{split}
\end{equation*}
Then by Lemma \ref{typical} we can
bound the first term of the right hand side, such that together with
the independence of all involved random variables we end up with
\begin{equation}\label{eq:46}
\begin{split}
 &\mathbb{E}_{\mathcal{X}}  (  W_t^{\otimes n}((D_j(\mathcal{X}) )^{c}| X_{jl}^{(t)} )) \\
 & \leq (n+1)^{|A||B|}\cdot 2^{-nc\delta^2} \\
 & + \sum_{\substack{ j'\in[J_n]\\ j'\neq j }}\sum_{r\in \theta}\sum_{k\in[L_{n,r}]} \mathbb{E}_{X_{j'k}^{(r)}}
      \mathbb{E}_{X_{jl}^{(t)}}  W_t^{\otimes n}( \mathcal{T}_{W_r,\delta}^n  (X_{j'k}^{(r)})|X_{jl}^{(t)}     ).
\end{split}
\end{equation} 
We shall find now for $j'\neq j$ by the same reasoning as in  \eqref{eq:22} and
\eqref{eq:23} an upper bound on
\begin{equation*}
\begin{split}
\mathbb{E}_{X_{jl}^{(t)}} &  W_t^{\otimes n}( \mathcal{T}_{W_r,\delta}^n 
(X_{j'k}^{(r)})|X_{jl}^{(t)}     ) \\
&\leq  \frac{q_t^{\otimes n}(\mathcal{T}_{W_r,\delta}^n 
(X_{j'k}^{(r)})  )  }{ p_t^{\otimes n}(\mathcal{T}_{p_t,\delta}^n)}\\
&\le \frac{(n+1)^{|A||B|}}{1-(n+1)^{|A|}\cdot 2^{-nc\delta^2}}\cdot 2^{-n (I(p_r,W_r) -f(\delta))}
\end{split}
\end{equation*}
for all $r,t\in \theta$, all $j'\neq j$, and all $l\in [L_{n,t}], k\in [L_{n,r}]$. Now by defining
$\nu_n(\delta):=(n+1)^{|A||B|}\cdot 2^{-nc\delta^2}$ and $\mu_n(\delta):=1-(n+1)^{|A|}\cdot
2^{-nc\delta^2}$ thus  for each $t\in \theta$, $l\in [L_{n,t}]$, and  $j\in[J_n]$ \eqref{eq:46}
and the last inequality leads to
\begin{equation*}
\begin{split}
&\mathbb{E}_{\mathcal{X}}  (  W_t^{\otimes n}((D_j(\mathcal{X}) )^{c}| X_{jl}^{(t)} )) \\
& \le \nu_n(\delta)
+ \frac{(n+1)^{|A||B|}}{\mu_n(\delta)}J_n \sum_{r\in \theta} L_{n,r}  2^{-n (I(p_r,W_r) -f(\delta))} \\
&\le \nu_n(\delta)  
+ \frac{(n+1)^{|A||B|}}{\mu_n(\delta)} T  J_n \cdot 
2^{-n (\min_{r\in\theta}( I(p_r,W_r)- \max_s I(p_r,V_s) ) -f(\delta)-\frac{\tau}{4}   )} \\
&\le \nu_n(\delta)
 + \frac{(n+1)^{|A||B|}}{\mu_n(\delta)} T \cdot 2^{-n\frac{\tau}{2}} 
\end{split}
\end{equation*}
where we have used the definitions of $J_n$ and $L_{n,r}$ in \eqref{eq:44}, \eqref{eq:45} and we have chosen
$\delta>0$ small enough to ensure that $\tau -f(\delta)-\frac{\tau}{4}\ge \frac{\tau}{2} $. Defining
$a=a(\delta,\tau):=\frac{\min\{ c\delta^2,  \frac{\tau}{2}  \}}{2}$ we can find
$n(\delta,\tau,|A|,|B|)\in\mathbb{N}$ such that  for all $n\ge n(\delta,\tau,|A|,|B|) $ we end in
\begin{equation*}
\mathbb{E}_{\mathcal{X}}  (\lambda^{(t)}_n(\mathcal{X})) \le T\cdot 2^{-na}.
\end{equation*} 
for any $t\in \theta$.
To give a bound on the average probability of error we define the event $\iota_0(t)$ for any $t \in \theta$
as in \eqref{eq:25} and the event 
\begin{equation*}
\iota:=\bigcap_{t\in\theta}\bigcap_{s\in \mathcal{S}}\bigcap_{k=0}^{J_n}\iota_k(t,s)
\end{equation*}
differs from \eqref{eq:27} only by the intersection of the unknown channel states $s \in
\mathcal{S}$. Thus we can conclude that
\begin{eqnarray*}
\textrm{Pr}\{ \iota^c\} 
&\leq & S\cdot T \cdot \epsilon +S \cdot T^{\frac{3}{2}}\cdot 2^{-n \frac{a}{2}} \nonumber\\
&\le &  S \cdot T^2 \cdot 2^{-nc''}
\end{eqnarray*}
holds for a suitable positive constant $c''>0$ and all sufficiently large $n\in\mathbb{N}$, and we have shown that for each $t\in \theta$ there exist realisations 
$\{ (x^{(t)}_{jl})_{j\in[J_n], l\in [L_{n,t}]}: t\in\theta\}\in \iota$ of 
$\mathcal{X}$.
By the same reasoning as in \eqref{eq:28} we get for any channel realisation $t\in\theta$ to the
legitimate receiver
\begin{equation*}
\left\| \frac{1}{L_{n,t}} \sum_{l=1}^{L_{n,t}} V^n_s(\cdot|x^{(t)}_{jl})- \Theta_s(\cdot)\right\| \leq 5 \epsilon
\end{equation*}
for each of the unknown channels $s \in \mathcal{S}$ to the eavesdropper.  Now, because for any $t \in
\theta$ we have a different codeword set $\{x^{(t)}_{jl}\}$, we slightly change the definition in
\eqref{eq:31} to
\begin{equation*}
 \hat{V}^n_{(s,t)}(z^n|(j,l)):= V^n_s(z^n|x^{(t)}_{jl})
\end{equation*}
and accordingly to $\hat{V}^n_{(s,t),j}$ and $\bar{V}^n_{(s,t)}$ in \eqref{eq:32}, \eqref{eq:33} in
that way, that these distributions are defined separately for each codeword set $t \in \theta$. Thus we get, that
\begin{equation*}
\|\hat{V}^n_{(s,t),j}-\bar{V}^n_{(s,t)} \| \leq 10\epsilon
\end{equation*}
is fulfilled for all $s \in \mathcal{S}$ for each individual channel $t \in \theta$ to the legitimate receiver.\\
Hence, using the same expurgation scheme as in the previous sections we have shown that there is a
sequence of $(n,\tilde{J}_n)$ codes for which
\begin{equation*}
\max_{t \in \theta} \, \max_{j \in [\tilde{J}_n]} \frac{1}{L_{n,t}} \sum_{l \in [L_{n,t}]} W^{\otimes n}_t(D^c_j |
x^{(t)}_{jl}) \leq T^{\frac{1}{4}} \cdot 2^{-n\frac{a'}{2}}
\end{equation*}
holds for sufficiently large $n\in\mathbb{N}$, and the strong secrecy level is fulfilled for every channel
$t\in\theta$ by
\begin{equation*}
I(J;Z^n_s) \leq -10 \epsilon\log(10\epsilon)+10n\epsilon \log|C|
\end{equation*}
which tends to zero for $n \to \infty$ for all channels $s \in \mathcal{S}$ to the eavesdropper. Thus we
have shown that
\begin{equation*}
R_S = \min_{t \in \theta} \max_{p\in\mathcal{P}(A)}(I(p,W_t)-\max_{s:(s,t) \in \mathcal{S} \times \theta}I(p,V_s))
\end{equation*}
is an achievable  secrecy rate for the compound wiretap channel $\cup_{t \in \theta}\mathfrak{W}_t$ in the case
where the channel state to the legitimate receiver is known at the transmitter.
\end{proof}
\emph{Remark.} By considering the converse of Theorem \ref{CSI-t}, we get for each $t \in \theta$ possible
channel realisations $\mathfrak{W}_t:=\{(W_t,V_s): s=1,\ldots S\}$. Then we can describe the compound
channel as $\mathfrak{W}=\cup_{t \in \theta}\mathfrak{W}_t$. In accordance to the case of no CSI  for each
$t \in \theta$ we obtain that 
\begin{equation*}
C_S(\mathfrak{W}_t) = \lim_{n\to\infty}\frac{1}{n}\max_{U\to X^n\to Y_t^n Z_s^n}(I(U;Y_t^n)-
\sup_{s\in \mathcal{S}}I(U;Z_s^n)   ).
\end{equation*}
\begin{proposition}
The secrecy capacity of the compound wiretap channel in the case where only the channel state to the
legitimate receiver is known at the transmitter $C_{S,CSI_t}(\mathfrak{W})$ is given by 
\begin{equation*}
C_{S,CSI_t}(\mathfrak{W}) =  \inf_{t \in \theta} C_S(\mathfrak{W_t}).
\end{equation*}
\end{proposition}
\smallskip
Now, additionally let us assume that each $V_s$ is a degraded version of every $W_t$ for $s\in
\mathcal{S}$ and $ t\in \theta$. Then as shown in
Lemma \ref{concavity} $I(X;Y_t|Z_s)$ is a concave function with respect to the input distribution $p_X=p$.
In particular this still holds for $\min_{s \in \mathcal{S}} I(X;Y_t|Z_s)$. Now because the random
variables $X,Y_t,Z_s$ form a Markov chain for all $t \in \theta$ and $s \in \mathcal{S}$ and 
\begin{equation*}
\min_{s \in \mathcal{S}}I(X;Y_t|Z_s)=I(X,Y_t)-\max_{s \in \mathcal{S}}I(X,Z_s),
\end{equation*}
for any $t \in \theta$ we get the upper bound on the secrecy rate as the secrecy capacity of a single
channel $W_t$ with $S$ channels to the eavesdropper. Then we can conclude
\begin{proposition}
The secrecy capacity of the channel where only the channel states to the legitimate receiver are known
and the channels to the eavesdropper are degraded versions of those to the legitimate receiver is given
by
\begin{equation*}
 C_{S,CSI_t}(\mathfrak{W}) = \min_{t\in \theta} \max_{p\in\mathcal{P}(A)}(I(p,W_t)-\max_{s \in
   \mathcal{S} }I(p,V_s)).
\end{equation*}
\end{proposition}

\subsection{Compound wiretap channel with $C_S=C_{S,CSI}$}

Let $\mathfrak{W}:=\{ W_t,V_s: t=1,\ldots T, s=1,\ldots S\}$ with $S\neq T$ and the pair $(t,s)$ unknown
to both the transmitter and the legitimate receiver. In addition let us assume that
\begin{equation}\label{eq:47}
\exists \, \hat{t} \in \theta \; \forall \, t \in \theta \; \exists \, U_t: \quad
W_{\hat{t}}=U_t W_t,
\end{equation}
which means that  $W_{\hat{t}}$ is a degraded version of all channel $W_t$ with $t \neq \hat{t}$. We
further assume that
\begin{equation}\label{eq:48}
\exists \, \hat{s} \in \mathcal{S} \; \forall \, s \in \mathcal{S}\;  \exists \, \hat{U}_s: \quad
 V_s = \hat{U}_s V_{\hat{s}},
\end{equation}
which means that all $V_s$ with $s \neq \hat{s}$ are degraded versions of $V_{\hat{s}}$. Then we can show
that the capacity of this channel equals the capacity of the same channel with CSI at the transmitter, e.g.
\begin{equation*}
C_{S}(\mathfrak{W}) = C_{S,CSI}(\mathfrak{W}).
\end{equation*}
First, by Theorem \ref{low-bound} it holds that
\begin{equation}\label{eq:49}
C_S(\mathfrak{W}) \geq \max_{M\rightarrow X \rightarrow (Y_t Z_s)} \min_{(t,s)} \,  (I(M,Y_t)-I(M,Z_s)),
\end{equation}
where $M$ is an auxiliary random variable, such that $M, X, (Y_t,Z_s)$ form a Markov chain $M\rightarrow X
\rightarrow (Y_t Z_s)$ in this order. Now let 
\begin{equation*}
p^*_{MX}=\arg \max_{M\rightarrow X \rightarrow (Y_{\hat{t}} Z_{\hat{s}})}(I(M,Y_{\hat{t}}) - I(M, Z_{\hat{s}}))
\end{equation*}
the joint distribution of $M$ and $X$ that achieves capacity for the single wiretap channel
$(W_{\hat{t}},V_{\hat{s}})$. Because the capacity of the compound wiretap channel $\mathfrak{W}$ is less
than or equal the capacity of each single channel we obtain
\begin{eqnarray}
C_{S,CSI} (\mathfrak{W}) &\leq& I(p^*_M, W_{\hat{t}}\cdot P^*_{X|M})- I(p^*_M,V_{\hat{s}}\cdot P^*_{X|M})
=C_S(W_{\hat{t}},V_{\hat{s}}) \nonumber \\
&\leq& I(p^*_M,U_t (W_t \cdot P^*_{X|M})) - I(p^*_M,\hat{U}_s (V_{\hat{s}}\cdot P^*_{X|M})) \nonumber \\
&\leq&I(p^*_M,W_t \cdot P^*_{X|M}) - I(p^*_M,V_s \cdot P^*_{X|M}) \label{eq:upper}
\end{eqnarray}
for all $(s,t) \in \mathcal{S} \times \theta$ because of \eqref{eq:47}, \eqref{eq:48}. Then by the last inequality it follows that
\begin{eqnarray*}
 I(p^*_M, W_{\hat{t}}\cdot P^*_{X|M})- I(p^*_M,V_{\hat{s}} \cdot P^*_{X|M}) &=& \min_{(s,t)}(I(p^*_M,W_t
 \cdot P^*_{X|M}) - I(p^*_M,V_s \cdot P^*_{X|M}))\\
&\leq& \max_{M\rightarrow X \rightarrow (Y_t Z_s)} \min_{(t,s)} \,  (I(M,Y_t)-I(M,Z_s))
\end{eqnarray*}
Now taking into account \eqref{eq:49} and \eqref{eq:upper} we end in
\begin{equation*}
C_{S,CSI}(\mathfrak{W}) \leq C_S(\mathfrak{W})
\end{equation*}
and therewith for this channel the lower bound of the capacity  without CSI matches the capacity of
the compound wiretap channel with CSI.

\section{Examples}\label{sec:ex}

In this section we provide some examples which display some striking features of compound wiretap
channels as opposed to the usual compound channels. Our first example shows clearly that for compound
wiretap channels with CSI at the transmitter the strategy of sending both the message and the
randomisation parameter does not work. The second one demonstrates that even in the case where the sets
of channels to the legitimate receiver and the eavesdropper both are convex, we can have
\begin{equation*}
C_{S,CSI}(\mathfrak{W})>0 \; \textrm{and} \; C_S(\mathfrak{W})=0,
\end{equation*}
as opposed to the case of the usual compound channel where the Minimax-Theorem applies.

In the following we use some simple facts which we state here without proof.\\ 
\emph{Fact 1.} The binary entropy function 
\[ h(x):=-x\log x -(1-x)\log (1-x), \quad x\in [0,1], \] 
is strictly increasing on $[0,\frac{1}{2}]$.\\
\emph{Fact 2.} Let $\eta\in [0,1]$ and set 
\[ D_{\eta}:=   
\begin{pmatrix}
1-\eta & \eta\\
\eta & 1-\eta 
\end{pmatrix}
.\]
Then for every $\tau,\tau'\in [0,1]$ it follows that
\[D_{\tau}D_{\tau'}=D_{\tau+\tau'-2\tau\tau'}.  \]
Moreover, if $\tau, \tau'\in (0,\frac{1}{2})$ then 
\[ \tau+\tau'-2\tau\tau'\in (0,\frac{1}{2})\ \textrm{and} \ \tau+\tau'-2\tau\tau' > \tau, \tau'. \]
\emph{Fact 3.} For $\tau, t\in [0,1] $
\[ (1-t)D_0+ tD_{\tau}=D_{t\tau}. \]

\subsection{Example 1}\label{sec:example1}

Consider a compound wiretap channel $\mathfrak{W}=\{(W_t,V_t):t=0,1\}$ in the case of CSI at the
transmitter. First we define the channels to the legitimate receiver and to the eavesdropper
for $t=0$ by
\begin{equation*}
W_0=D_\eta, \; \eta \in [0,\frac{1}{2}), \; \eta \approx 0, \quad\quad V_0:=D_\tau W_0, \;  \tau \in [0,\frac{1}{2}), \; \tau \approx 0,
\end{equation*}
and for $t=1$, $\hat\tau \in (0,1/2]$
\begin{equation*}
W_1:= D_{\hat{\tau}} V_0=D_{\hat{\tau}}D_\tau W_0, \quad\quad
V_1:= \begin{pmatrix}
\frac{1}{2} & \frac{1}{2}\\
\frac{1}{2} & \frac{1}{2}
\end{pmatrix}.
\end{equation*}
Hence $V_0$ and $W_1$ are degraded versions of $W_0$ and
\begin{equation*}
I(p,V_1)=0, \quad \forall p \in \mathcal{P}(A)
\end{equation*}
by definition of $V_1$. Now for every $p \in \mathcal{P}(A)$ we can choose $\tau$ small enough, that
\begin{equation*}
I(p,W_0)-I(p,V_0) < I(p,W_1).
\end{equation*}
Now with $p_0=(\frac{1}{2},\frac{1}{2})$, $\nu>0$ and because we have CSI at the transmitter we have by the defining equations \eqref{eq:13} and
\eqref{eq:14}  
\begin{eqnarray*}
J_n&=& 2^{n[I(p_0,W_0)-I(p_0,V_0))-\nu]} \\
L_{n,0}&=& 2^{n[I(p_0,V_0)+\frac{\nu}{4}]}
\end{eqnarray*}
such that we obtain
\begin{equation*}
J_nL_{n,0} =  2^{n[I(p_0,W_0)-\frac{3\nu}{4}]}.
\end{equation*}
But for $\hat{\tau}$ close to $1/2$ it holds then that
\begin{equation*}
I(p_0,W_0)-\frac{3\nu}{4} >  I(p_0,W_1)
                                    = \max_{p\in\mathcal{P}(A)} I(p,W_1)=C_{CSI}\{W_0,W_1\},
\end{equation*}
where $C_{CSI}\{W_0,W_1\}$ is the capacity of a compound channel with CSI at the transmitter. Hence we
have shown, that we can achieve reliable transmission of the message $j \in [J_n]$, but identifying both
the message  and the randomizing indices is not possible for all pairs $j \in [J_n]$ and $l \in
[L_{n,t}]$. This is in contrast to the case where we have only one channel to both the legitimate
receiver and the eavesdropper (cf. \cite{devetak}, \cite{csis96}).

\subsection{Example 2}\label{sec:example2}

Now, for $\eta, \tau \in (0,\frac{1}{2})$ we set
\begin{eqnarray*}
W_0= D_\eta,  \quad V_0:=D_\tau W_0=D_{\eta+\tau-2\eta\tau},\\ \nonumber \\
W_1:= D_\tau V_0=D_{2\tau-2\tau^2}W_0,  \quad V_1:=D_\tau W_1.
\end{eqnarray*}
Notice that $V_0$ is a degraded version of $W_0$, $W_1$ of $V_0$, and $V_1$ of $W_1$. 
Next we define for $t \in [0,1]$
\begin{eqnarray}
W_t & := & (1-t)W_0+ tW_1 \nonumber \\
       &  = & \big[ (1-t)D_0 + t D_{2\tau-\tau^2}] W_0, \label{eq:51}\\ \nonumber\\
V_t  & := & (1-t)V_0 + tV_1 \nonumber \\
       & = & D_\tau \big[ (1-t)D_0 + t  D_{2\tau-2\tau^2} \big] W_0 \nonumber \\
       & = & D_\tau W_t \label{eq:52}
\end{eqnarray}
By the definition, the set of
channels to the legitimate receiver $\{W_t\}$ and the set of channels to the eavesdropper $\{V_t\}$ both
are convex. Nevertheless we will show now, that for the compound  wiretap channel
$\mathfrak{W}:=\{(W_t,V_t):t \in [0,1]\}$ we have
 \[ C_{S,CSI}(\mathfrak{W})>0,\quad C_S(\mathfrak{W})=0 . \]
To this end, note that by (\ref{eq:51}), fact 3, and fact 2 we have
\begin{equation*}
 W_t= D_{t(2\tau-\tau^2)}D_\eta=D_{f(t,\eta,\tau)}
\end{equation*} 
with
\begin{equation}\label{eq:53} 
f(t,\eta,\tau):= \eta +t(2\tau-2\tau^2)-2\eta t (2\tau-2\tau^2)\in (0,\frac{1}{2}). 
 \end{equation}
Similarly from (\ref{eq:52}) and fact 2 we obtain
\begin{equation*}
V_t= D_\tau D_{f(t,\eta,\tau)}
   = D_{\tau+f(t,\eta,\tau)-2\tau f(t,\eta,\tau)}
\end{equation*} 
Additionally from \eqref{eq:53} and fact 2 we get
\begin{eqnarray}
 \tau+f(t,\eta,\tau)-2\tau f(t,\eta,\tau) \in (0,\frac{1}{2}) \quad \textrm{and} \nonumber \\
 \tau+f(t,\eta,\tau)-2\tau f(t,\eta,\tau)>f(t,\eta,\tau).\label{eq:54}
\end{eqnarray} 
Taking $p=(1/2,1/2)$ we obtain for every $t\in [0,1]$
\begin{equation*}
 I(p,W_t)-  I(p,V_t)
 =  h(  \tau+f(t,\eta,\tau)-2\tau f(t,\eta,\tau))-h( f(t,\eta,\tau))>0
\end{equation*} 
where the last inequality follows from fact 1 and (\ref{eq:54}). Thus we have shown that
\[ C_{S,CSI}(\mathfrak{W})>0 \]
holds by Theorem \ref{CSI-code}.

In order to show that $C_S(\mathfrak{W})=0$, we have to employ our multiletter converse in the case of
no CSI, Proposition \ref{multiletter-converse}.
First, a simple algebra shows that for any $t\in [0,1]$
\[ V_t= ((1-t)D_0+t D_{2\tau-2\tau^2})V_0 \]
by \eqref{eq:52} and thus each $V_t$ is a degraded version of $V_0$. Let us now consider the Markov
chain $U \to X^n\to (Y_t^n,Z_t^n)$ where the transition from the random variable $U$ to $Y^n_t$ is
governed by $P_{Y^n_t|U}=V^{\otimes n}_t \cdot P_{X^n|U}$ for all $t \in [0,1]$. Then we obtain that each
$P_{Y^n_t|U}$ is a degraded version of $P_{Y^n_0|U}=V^{\otimes n}_0 \cdot P_{X^n|U}$, and the data
processing inequality implies that for each $n\in \mathbb{N}$
\begin{equation}\label{eq:55}
 \max_{t\in [0,1]}I(U, Y_t^{ n})= I(U,Y_0^{n})
\end{equation} 
for all distributions $P_{U X^n}$ that satisfy the Markov chain condition $U \to X^n\to (Y_t^n,Z_t^n)$.\\
On the other hand, since $W_1=D_{\tau}V_0$ we obtain for the matrix $P_{Z^n_1|U}=W^{\otimes n}_1 \cdot
P_{X^n|U}$ by the data processing inequality and (\ref{eq:55})
for all $n\in\mathbb{N}$
\begin{equation*}
 I(U,Z_1^{n})-\max_{t\in [0,1]}I(U,Y_t^{n})=I(U,Z_1^{n})-I(U,Y_0^{n})\le 0,
\end{equation*} 
for all $P_{UX^n}$. Then Proposition \ref{multiletter-converse} implies that
\[ C_S(\mathfrak{W})=0 \]
as desired.

\section*{Acknowledgment}
Support by the Deutsche Forschungsgemeinschaft (DFG) via projects 
BO 1734/16-1, BO 1734/20-1, and by the Bundesministerium f\"ur Bildung und Forschung (BMBF) via grant
01BQ1050 is gratefully acknowledged.

\newpage

\section*{References}

\bibliographystyle{apsrev4-1long}
\bibliography{references}
\newpage

%
%

%

\begin{figure*}[t!]

\setcaptionmargin{5mm}
\onelinecaptionstrue
\scalebox{1.5}{
\includegraphics{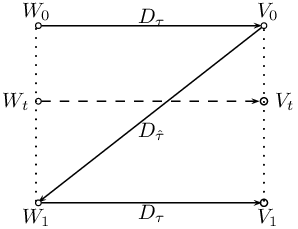}
}
\captionstyle{normal}
\caption{Compound wiretap channel $\mathfrak{W}:=\{(W_t,V_t):t \in [0,1]\}$ of Ex. 2}

\end{figure*} 

\end{document}